\documentclass[aps,prx,twocolumn,groupedaddress,amsmath,superscriptaddress]{revtex4-2}

\usepackage{amsfonts}
\usepackage{amsmath}
\usepackage[dvips]{graphicx}
\usepackage{color}
\usepackage{graphicx}
\usepackage{amssymb}
\usepackage{hyperref}
\usepackage{color}
\usepackage{mathrsfs}
\usepackage{cleveref}
\usepackage{lipsum}
\usepackage{isomath}
\usepackage{amsthm}
\usepackage{outlines}
\usepackage{epstopdf}
\usepackage{txfonts}
\usepackage{dsfont}
\usepackage{setspace}
\usepackage{multirow}
\usepackage{booktabs}
\usepackage{tabularx}
\usepackage[T1]{fontenc}
\usepackage{colortbl}
\usepackage{hhline}
\usepackage{xcolor}
\usepackage[normalem]{ulem}
\allowdisplaybreaks[4]

\makeatletter

\definecolor{nblue}{rgb}{0.3,0.3,1.0}
\definecolor{ngreen}{rgb}{0.2,0.7,0.2}
\definecolor{nred}{rgb}{0.9,0.1,0}
\definecolor{nblack}{rgb}{0,0,0}
\definecolor{nyellow}{rgb}{1.0,0.75,0.0}

\newcommand{\beq}{\begin{equation}}
\newcommand{\eeq}{\end{equation}}
\newcommand{\bqa}{\begin{eqnarray}}
\newcommand{\eqa}{\end{eqnarray}}

\newtheorem{proposition}{Proposition}
\newtheorem{definition}{Definition}
\newtheorem{lemma}{Lemma}
\newtheorem{result}{Result}
\newtheorem{remark}{Remark}

\crefname{lemma}{Lemma}{Lemmas}

\newcommand{\jt}[1]{{\color{ngreen} #1}}

\makeatletter

\usepackage{babel}

\begin{document}
	
\title{Necessary and Sufficient Condition for Randomness Certification from Incompatibility}
\author{Yi Li}
\affiliation{State Key Laboratory for Mesoscopic Physics, School of Physics, Frontiers Science Center for Nano-optoelectronics, $\&$ Collaborative Innovation Center of Quantum Matter, Peking University, Beijing 100871, China}
\address{Beijing Academy of Quantum Information Sciences, Beijing 100193, China}
\address{Naturwissenschaftlich-Technische Fakult\"{a}t, Universit\"{a}t Siegen, Walter-Flex-Stra{\ss}e 3, 57068 Siegen, Germany}
\author{Yu Xiang}
\affiliation{State Key Laboratory for Mesoscopic Physics, School of Physics, Frontiers Science Center for Nano-optoelectronics, $\&$ Collaborative Innovation Center of Quantum Matter, Peking University, Beijing 100871, China}
\author{Jordi Tura}
\affiliation{$\langle aQa^L\rangle$ Applied Quantum Algorithms Leiden, The Netherlands}
\affiliation{Instituut-Lorentz, Universiteit Leiden, P.O. Box 9506, 2300 RA Leiden, The Netherlands}
\author{Qiongyi He}
\email{qiongyihe@pku.edu.cn}
\affiliation{State Key Laboratory for Mesoscopic Physics, School of Physics, Frontiers Science Center for Nano-optoelectronics, $\&$ Collaborative Innovation Center of Quantum Matter, Peking University, Beijing 100871, China}
\affiliation{Collaborative Innovation Center of Extreme Optics, Shanxi University, Taiyuan, Shanxi 030006, China}
\affiliation{Peking University Yangtze Delta Institute of Optoelectronics, Nantong, Jiangsu 226010, China}	
\affiliation{Hefei National Laboratory, Hefei 230088, China}

\begin{abstract}
Quantum randomness can be certified from probabilistic behaviors demonstrating Bell nonlocality or Einstein-Podolsky-Rosen steering, leveraging outcomes from uncharacterized devices. 
However, such nonlocal correlations are not always sufficient for this task, necessitating the identification of required minimum quantum resources. In this work, we provide the necessary and sufficient condition for nonzero certifiable randomness in terms of measurement incompatibility and develop approaches to detect them. Firstly, we show that the steering-based randomness can be certified if and only if the correlations arise from a measurement compatibility structure that is not isomorphic to a hypergraph containing a star subgraph. 
In such a structure, the central measurement is individually compatible with the measurements at branch sites, precluding certifiable randomness in the central measurement outcomes. Subsequently, we generalize this result to the Bell scenario, proving that the violation of any chain inequality involving $m$ inputs and $d$ outputs rules out such a compatibility structure, thereby validating all chain inequalities as credible witnesses for randomness certification.
Our results point out the role of incompatibility structure in generating random numbers, offering a way to identify minimum quantum resources for the task.
\end{abstract}

\maketitle
\textit{Introduction.---} Quantum random numbers, which can be generated by performing a projective measurement on an appropriate quantum state, have been widely studied in fields ranging from quantum cryptography to fundamental science~\cite{rmp_2017_random,njp_2015_xiongfengma}. However, these protocols necessitate a perfect characterization of both the state and measurement, e.g., small alignment errors can cause an erroneous certification of entangled states~\cite{pra_2012_gisin}. 
To bypass the necessity of modeling both factors, the nonlocality-based randomness certification protocol is proposed through the violation of a Bell inequality~\cite{nature_2010_random}, where each party performs untrusted measurements~\cite{rmp_2014_nonlocality}. Considering the robustness to noise and experimental imperfections, this method is generalized to the scenario with the exhibition of Einstein-Podolsky-Rosen (EPR) steering~\cite{njp_2015_paul},
since it
only relies on the trustworthy measurements on one-side device~\cite{rmp_2020_steering,prl_2007_wiseman}.

\begin{figure}[t]
\centering
\includegraphics[width=0.43\textwidth]{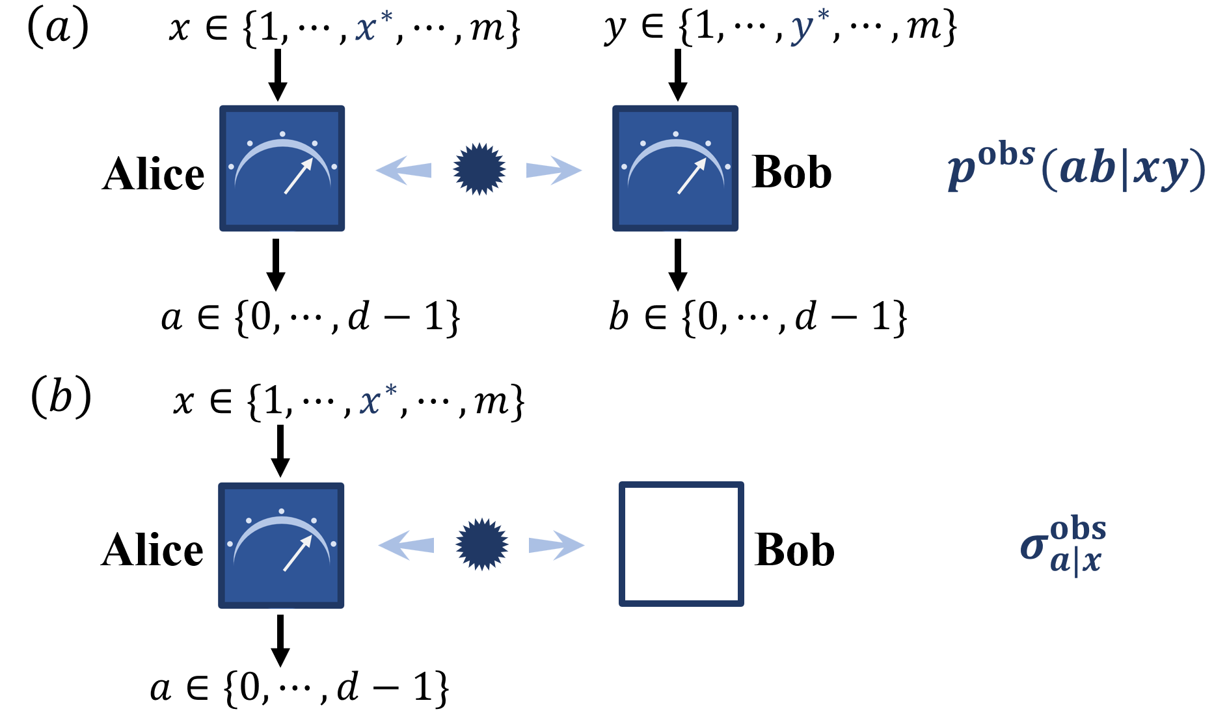}
\vspace{-0.1cm}
\caption{Randomness certification based on quantum nonlocality. (a) The Bell scenario, involves both Alice and Bob receiving $m$ inputs and producing $d$ outputs, where each of them implements untrusted local measurements. Then, their behavior is described by the joint probability distributions $ \{ p^{\rm obs}(ab|xy)\}$. (b) The steering scenario, which is similar to (a) but Bob performs trusted measurements, such that he holds the information of assemblage $\{\sigma_{a|x}^{\rm obs}\}$. In these scenarios, randomness on the outcomes of $x^*$ and $y^*$ can be certified based on the observed statistics. By constructing an incompatibility structure of the $m$ inputs, we clarify the necessary and sufficient condition for this task. Further, we prove that any chain inequality is a credible witness for randomness certification.}
\label{fig:scheme}
\vspace{-0.4cm}
\end{figure}

In both cases, each untrusted party, ``black box'' in Fig.~\ref{fig:scheme}, is required to perform a set of incompatible measurements, otherwise, the resulting statistics, i.e., the input-output data given by measurement devices, cannot express Bell nonlocality or EPR steering~\cite{prl_2014_brunner,prl_2014_otfried,prl_2015_otfried}. 
Here, the measurement incompatibility represents the inability to perform measurements simultaneously~\cite{rmp_2023_otfried}, which has shown advantages in the task of quantum state discrimination~\cite{prl_2019_paul, prl_2019_roope,prl_2020_eric,quantum_2019_oszmaniec,arXiv_2023_peizhang,prl_2019_carmeli}. Intuitively, if a set of measurements is incompatible, then not all measurement outcomes are predictable, indicating that these outcomes are random. 
Therefore, Bell nonlocality and EPR steering enable to witness measurement incompatibility~\cite{prl_2019_quintino,prr_2021_shiliangchen,prl_2009_wolf,prl_2016_shinliangchen}, allowing the certification of randomness in measurement outcomes without making assumptions about the internal mechanisms of the untrusted devices~\cite{np_2021_ustc,nphy_2021_random,prl_2018_christian,prl_2018_ustc,prl_2021_ustc,qst_2017_paul,science_2018_jianweiwang,prl_2019_shumingcheng,nature_2010_random,pra_2022_exp}.

In these approaches, randomness can be certified based on the observed statistics. 
Some previous works~\cite{nature_2010_random,prl_2012_acin,njp_2015_paul,njp_2014_the,quantum_2018_multipartite,pra_2013_multipartite,pra_2022_the,prl_2019_shumingcheng,pra_2022_exp,arXiv_2024_USTC,PRL_2024_yi} showed that the observation of nonlocal correlations is sufficient to certify randomness when the untrusted party  implements measurements chosen from only two different settings. 
However, in cases involving more measurement settings where a compatibility structure can generally be described by a hypergraph~\cite{prl_2019_quintino}, not all observed statistics that express EPR steering or Bell nonlocality contribute to certifying randomness~\cite{PRL_2024_yi,arXiv_2024_USTC,arXiv_2024_stefano}. For example, there exist some Bell inequalities that cannot certify randomness even when they are maximally violated~\cite{arXiv_2024_stefano}. 
Meanwhile, some pure states with very weak entanglement can still be used to certify maximal randomness~\cite{prl_2012_acin,prl_2018_paul}. 
Therefore, three questions arise: What are the necessary and sufficient conditions for observed statistics to certify randomness? 
For which Bell inequalities can randomness still be certified when they are only slightly violated? How to determine whether a quantum state itself is useful for certifying randomness in these approaches? Answering these questions is important for fully leveraging quantum resources in the task of randomness certification.

In this work, we find that randomness can be certified in the steering scenario iff (if and only if) the correlations arise from a measurement compatibility structure that is not isomorphic to a hypergraph containing a star subgraph. In such a structure as shown in Fig.~\ref{fig:structure}~(a), the central measurement is individually compatible with the measurements at branch sites, precluding certifiable randomness in the central measurement outcomes. 
Further, we extend this finding to the Bell scenario where randomness can be certified iff at least one of Alice's or Bob's measurement compatibility structures does not contain a star subgraph. 
As a result, we prove that the violation of any chain inequality~\cite{prl_2002_cglmp,prl_2017_satwap,arXiv_2022_mengyaohu,prl_2006_chain} with $m$ inputs and $d$ outputs rules out such a compatibility structure, thereby all chain inequalities can act as truthful witnesses for randomness certification. 
Based on these results, we assert that a state useful for certifying randomness is determined by whether it can express steerability when restricted to two measurement settings.

\textit{Nonlocality-based randomness certification---} In Fig.~\ref{fig:scheme}~(a), Alice and Bob are situated separately and share a quantum state $\rho_{AB}$. Both of them accept $m$ inputs, labeled by $x,y\in\left\{1,\cdots,m\right\}$, and produce $d$ outputs, labeled by $a,b\in[d]$, where $[d] \coloneqq \left\{ 0,\cdots,d-1 \right\}$. A set of resulting joint probability distributions $\{p^{\rm obs} (ab|xy)\}$ exhibiting Bell nonlocality cannot be
described by local-hidden variable (LHV) model:
$\sum_{\lambda }p(\lambda ) p_A(a|x,\lambda)p_B(b|y,\lambda)$, $\forall a,b,x,y$, 
where $\lambda$ are the hidden variables, and $p(\lambda)$, $p_A(a|x,\lambda)$, $p_B(b|y,\lambda)$ are probability distributions. Based on this, randomness of outcomes of measurements $x^*$ and $y^*$ can be certified by considering an eavesdropper (Eve) who attempts to guess both outcomes with a guessing probability of~\cite{njp_2014_the} 
\begin{equation}
    \begin{aligned}
    \label{eq:definition_nl_rand}
P_{\text {guess }}^{\rm NL}\left(x^*,y^*\right)& \coloneqq \max _{\left\{p^{ee'}(ab\mid xy)\right\}} \sum_{ee'} p^{ee'}\left(a=e,b=e' \mid x^*,y^*\right)\\
\text { s.t. } & ~~\sum_{ee'} p^{ee'}(ab|xy)=p^{\text{obs}}(ab|xy), \quad \forall a, b, x, y,\\
& ~~p^{ee'}(ab|xy)\in Q, \quad  \forall e, e'.
\end{aligned}
\end{equation}
Here $Q$ represents the quantum set~\cite{rmp_2014_nonlocality}, and $ \{ p^{ee'}(ab|xy) \} $ denotes a decomposition of the observed statistics, where $e,e' \in [d]$ are guesses from Eve. Thus, randomness of these outcomes is quantified by the min-entropy $H_{\min}^{\rm NL}\left(x^*,y^*\right) \coloneqq -\log_2P_{\text{guess}}^{\rm NL}\left(x^*,y^*\right)$. Additionally, by simply modifying the first constraint in Eq.~\eqref{eq:definition_nl_rand}, randomness can also be certified based on the violation of a Bell inequality~\cite{science_2018_jianweiwang,nature_2010_random}. 
Clearly, since a behavior satisfying the LHV model can be rewritten as a mixture of deterministic behaviors~\cite{prl_1982_fine}, this behavior gives $\min_{x^*,y^*}P_g^{\rm NL}(x^*,y^*) = 1$, hence, zero nonlocality-based randomness.

\textit{Steering-based randomness certification---} For the analogous scenario shown in Fig.~\ref{fig:scheme}~(b), Bob's measurement device is now trusted but Alice's remains untrusted. This allows Bob to perform tomographic measurements and acquire knowledge of a set of conditional states $\{\sigma_{a|x}^{\text{obs}}\}$, which is also called assemblage. Based on this observed assemblage one can certify randomness on the outcomes of measurement $x^*$ by means of a semidefinite program~\cite{njp_2015_paul}:
\begin{equation}
    \begin{aligned}
    \label{eq:definition_steer_rand}
P_{\text {guess }}^{\rm S}\left(x^*\right)& \coloneqq \max _{\left\{\sigma_{a \mid x}^e\right\}} \sum_e \operatorname{Tr}\left(\sigma_{a=e \mid x^*}^e\right) \\
\text { s.t. } & ~~\sum_e \sigma_{a \mid x}^e=\sigma_{a \mid x}^{\text {obs }}, ~ \forall a, x, \\
& ~~\sum_a \sigma_{a \mid x}^e=\sum_a \sigma_{a \mid x^{\prime}}^e,~ \forall e, x \neq x^{\prime}, \\
& ~~ \sigma_{a \mid x}^e \succeq 0,~ \forall a, x, e,
\end{aligned}
\end{equation}
where the no-signaling constraint is equivalent to the restriction of quantum realization in bipartite scenario~\cite{pr_1957_ghjw,hpa_1989_ghjw,pla_1993_ghjw,mpcps_1936_schro}. 
Then, the steering-based randomness is quantified by $H_{\min}^{\rm S}(x^*) = -\log_2 P_{\text{guess}}^{\rm S}(x^*)$. An assemblage admitting the local-hidden state (LHS) model: $\sum_{\lambda}p_A(a|x,\lambda)\sigma_\lambda$, $\forall a,x$,
also gives $\min_{x^*}P_g^{\rm S}(x^*) = 1$, where the hidden states satisfy $\sigma_\lambda \succeq 0$ and $\sum_\lambda {\rm Tr}\left( \sigma_{\lambda} \right) = 1$.

\textit{Compatibility structure of measurements---} 
In these scenarios, we assume that each untrusted party performs generalized measurements, corresponding to the so-called positive operator valued measures (POVMs). These are defined by a collection of positive semi-definite matrices $\{M_{a}\}$ that satisfy the normalization condition $\sum_{a}M_{a} = \mathbb{I}$. In the following, we denote Alice's (Bob's) measurements as $\{M_{a|x}\}$ ($\{M_{b|y}\}$). 
In this context, a set of POVMs $\{M_{a|x}\}$ is considered compatible if there exists a parent POVM denoted by $\{G_\lambda\}$ together with a set of conditional probabilities $\left\{p(a|x,\lambda)\right\}$ such that~\cite{rmp_2023_otfried}
\begin{equation}
    M_{a|x} = \sum_\lambda p(a|x,\lambda)G_\lambda,~~~\forall a,x.
\end{equation}
Then, for a set containing more than two measurements, its compatibility structure is represented by a hypergraph 
denoted $\mathcal{C} \equiv \left[C_1,C_2,\cdots,C_k\right]$, where each vertex corresponds to a measurement, and each hyperedge $C_i$ indicates that the corresponding subset of measurements is compatible~\cite{prl_2019_quintino}. Illustrative examples are shown in Fig.~\ref{fig:structure}. In this work, we denote the compatibility structures of Alice's and Bob's measurements $\{M_{a|x}\}$ and $\{M_{b|y}\}$ as $\mathcal{C}_A$ and $\mathcal{C}_B$.
\begin{figure}[t]
\centering
\includegraphics[width=0.47\textwidth]{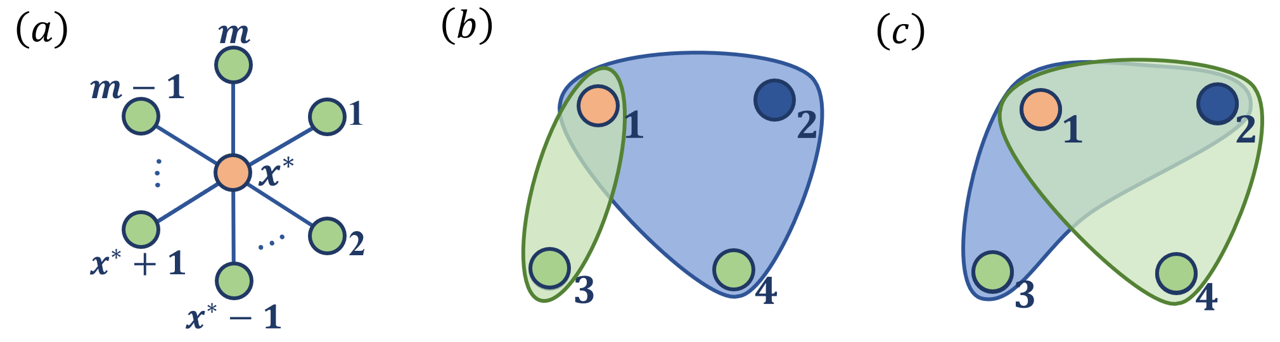}
\vspace{-0.3cm}
\caption{The compatibility structures of measurements. (a)~A star graph $K_{1,m-1}\left(x^*\right)$. This means the measurement $x^*$ is compatible with every other measurement. (b)~A hypergraph indicating two subsets of compatible measurements: $\left\{1, 2, 4\right\}$ and $\left\{1,3\right\}$. (c)~The set of measurements $\{1,2\}$, together with every other measurement, is compatible.}
\label{fig:structure}
\vspace{-0.3cm}
\end{figure}

In what follows we will show that the compatibility structure containing a particular star subgraph $K_{1,m-1}(x^*)$, as shown in Fig.~\ref{fig:structure}~(a), provides zero certifiable randomness.

\begin{proposition}\label{proposition 1}
\normalfont Given an assemblage $\{\sigma_{a|x}^{\rm obs}\}$, the guessing probability satisfies $P_g^{\rm S}(x^*) < 1$ iff it lies outside the set $\mathcal{R}^{\rm S}_{x^*}$ for any $x^*\in\{1,\cdots,m\}$, where $\mathcal{R}^{\rm S}_{x^*}$ is defined as:
\begin{equation}
    \left\{ \sigma_{a|x}: 
    \begin{array}{ll}
        &\sigma_{a|x} = {\rm Tr}_A\left[\left(M_{a|x}\otimes \mathbb{I}^B \right)\rho_{AB}\right], \mathcal{C}_{A} \supseteq K_{1,m-1} (x^*),\\
        &M_{a|x}\succeq 0, \sum_a M_{a|x} = \mathbb{I}^A, \rho_{AB}\succeq 0, {\rm Tr} (\rho_{AB}) = 1, 
    \end{array} \right\}.
\end{equation}
\end{proposition}
\textit{Sketch of Proof.} Firstly, for the assemblage inside the set $\mathcal{R}^{\rm S}_{x^*}$, since $\mathcal{C}_{A} \supseteq K_{1,m-1} (x^*)$, the assemblage lies in a ``$K_{1,m-1} (x^*)$-partially-unsteerable'' set, 
meaning that the ensemble $\{\sigma_{a|x}^{\rm obs}\}_{a,x\in\{x^*,x'\}}$ admits the LHS model for any edge $(x^*,x')$ in the graph $ K_{1,m-1}(x^*) $. According to this, we construct a solution $\left\{\xi_{a|x}^e\right\}$ to Eq.~\eqref{eq:definition_steer_rand} giving $P_{g}^{\rm S}(x^*) = 1$. 
Conversely, if an assemblage yields $P_g^{S}(x^*) = 1$, it can always be explained by a scenario where $\mathcal{C}_A \supseteq K_{1,m-1}(x^*)$~\cite{PRL_2024_yi}, which also implies that the ``$K_{1,m-1} (x^*)$-partially-unsteerable'' set and $\mathcal{R}_{x^*}^{\rm S}$ are the same. Therefore, all the assemblages outside $\mathcal{R}^{\rm S}_{x^*}$ give $P_g^{S}(x^*) < 1$, resulting in nonzero randomness. The detailed proof can be found in Appendix~A. 

In the operational interpretation, if an assemblage $\sigma_{a|x}^{\rm obs}\in\mathcal{R}^{\rm S}_{x^*}$, Proposition~\ref{proposition 1} implies an explanation of
\begin{equation}
    \begin{aligned}
    \label{eq:proposition1}
        \sigma_{a|x}^{\text{obs}} = \text{Tr}_{EA}\left[ \left( \mathcal{M}_{a|x}^{EA} \otimes \mathbb{I}^B\right)\rho_{EAB} \right],~~~\forall a,x,
    \end{aligned}
\end{equation}
where $\rho_{EAB} = \sum_{e} p(e)|e\rangle_E\langle e| \otimes \rho_{AB}^e$, $\mathcal{M}_{a|x^*}^{EA} = |a\rangle_E\langle a|\otimes \mathbb{I}^A$, and $\mathcal{M}_{a|x}^{EA} = \mathbb{I}^E\otimes \mathcal{M}_{a|x}^A$ for $x\neq x^*$. Here $\left\{|e\rangle\right\}_e$ are orthogonal states, $\{\rho_{AB}^e\}_e$ are normalized quantum states, and $\left\{\mathcal{M}_{a|x}^A\right\}$ are POVMs. 
This scenario considers the bipartition of $EA|B$, where the parties share the state $|e\rangle_E\langle e| \otimes \rho_{AB}^e$ with probability $p(e)$. For each $e$, the subsystem $EA$ always produces measurement outcome $e$ when the input is $x^*$. 

According to this, we then prove the following results. 

\begin{result}\label{result 2}
\normalfont For a given behavior $\{p^{\rm obs}(ab|xy)\}$, nonlocality-based randomness can be certified, i.e., $P_g^{\rm NL}(x^*,y^*) < 1$, iff the behavior lies outside the set $ \mathcal{R}^{\rm NL}_{x^*,y^*}\coloneqq\{ p(ab|xy): p(ab|xy)\in Q, \mathcal{C}_{A} \supseteq K_{1,m-1} (x^*), \mathcal{C}_{B} \supseteq K_{1,m-1} (y^*)\}$ for any $x^*,y^*\in\{1,\cdots,m\}$.
\end{result}

\begin{proof} 
Firstly, any behavior in the set $\mathcal{R}^{\rm NL}_{x^*,y^*}$ admits quantum realization, i.e., $p^{\rm obs} (ab|xy) = {\rm Tr}\left[\left(M_{a|x}\otimes M_{b|y} \right) \rho_{AB}\right]$. Since $\mathcal{C}_A \supseteq K_{1,m-1}\left(x^*\right)$, 
based on Eq.~\eqref{eq:proposition1}, the joint probability distributions being rewritten as
\begin{equation}
    \begin{aligned}
    \label{eq:jointpro}
        p^{\text{obs}}(ab|xy) = \text{Tr}\left[\left(\mathcal{M}_{a|x}^{EA} \otimes M_{b|y} \right) \rho_{EAB}\right],~~~\forall a,b,x,y.
    \end{aligned}
\end{equation}
Similarly, since $\mathcal{C}_B \supseteq K_{1,m-1}\left(y^*\right)$, the distributions can be described by
\begin{equation}
    \begin{aligned}
        p^{\text{obs}}(ab|xy) = \text{Tr}\left[ \left( \mathcal{M}_{a|x}^{EA} \otimes \mathcal{M}_{b|y}^{E'B} \right)\rho_{EE'AB} \right],~~~\forall a,b,x,y,
    \end{aligned}
\end{equation}
where $\rho_{EE'AB} = \sum_{e'} p(e')|e'\rangle_{E'}\langle e'| \otimes \rho_{EAB}^{e'}$, $\mathcal{M}_{b|y^*}^{E'B} = |b\rangle_{E'}\langle b|\otimes \mathbb{I}^{B}$ and $\mathcal{M}_{b|y}^{E'B} = \mathbb{I}^{E'}\otimes \mathcal{M}_{b|y}^B$ for $y \neq y^*$. Without loss of generality, Eve can hold the subsystems $EE'$, such that this underlying construction allows Eve to perfectly guess the outcomes of measurements $x^*$ and $y^*$, thereby providing $P_g^{\rm NL}(x^*,y^*) = 1$. Additionally, since a quantum behavior giving $P_g^{\rm NL}(x^*,y^*) = 1$ can always be produced in a scenario where $\mathcal{C}_A \supseteq K_{1,m-1}(x^*)$ and $\mathcal{C}_B \supseteq K_{1,m-1}(y^*)$ ~\cite{arXiv_2024_USTC}, then the behaviors outside the set $\mathcal{R}^{\rm NL}_{x^*,y^*}$ must give $P_g^{\rm NL}(x^*,y^*)<1$, which concludes the proof.
\end{proof}

According to Result~\ref{result 2}, not all Bell nonlocal behaviors that violate tight (or facet) Bell inequalities can certify randomness, as their violations do not necessarily exclude the specific structure of measurement compatibility. We then consider a class of chain inequalities~\cite{arXiv_2022_mengyaohu,prl_2006_chain,pra_2006_chain}:
\begin{equation}
    \begin{aligned}
        \label{eq:chain_manuscript}
        I_{d,m} = \sum_{k=0}^{d-1}\alpha_k \sum_{i=1}^m \left[p(A_i-B_i = k) + p(B_i-A_{i+1}=k)\right]\leq I_{d,m}^{C},
    \end{aligned}
\end{equation}
where $\alpha_k \in \mathbb{R} $, $A_{m+1} \equiv A_1 + 1$, and the Bell expression contains probabilities of $p(A_x - B_y = k)\coloneqq \sum_{j = 0}^{d - 1} p(j+k,j|xy)$. Here the arithmetic is taken \textit{modulo} $d$, and the classical bound is $I_{d,m}^{C} \coloneqq \max_{\boldsymbol{p} \in \mathcal{L}}I_{d,m}$, where $\mathcal{L}$ represents the local set. This class includes the Collins-Gisin-Linden-Massar-Popescu~\cite{prl_2002_cglmp} and the Salavrakos-Augusiak-Tura-Wittek-Ac\'{i}n-Pironio inequalities~\cite{prl_2017_satwap}, which have been used to experimentally certify nonlocality-based randomness~\cite{science_2018_jianweiwang,arXiv_2024_USTC}. In Result~\ref{result 3}, we prove that the violations of any chain inequality contribute to certifying randomness.

\begin{result}\label{result 3}
\normalfont For any $x^*\in\{1,\cdots,m\}$, if $\mathcal{C}_A\supseteq K_{1,m-1}(x^*)$, the corresponding distributions $\{p^{\text{obs}}(ab|xy)\in Q\}$ cannot violate any chain inequality, and thus any violation of a chain inequality is sufficient to certify nonlocality-based randomness.
\end{result}

\textit{Sketch of Proof.} Based on Eq.~\eqref{eq:proposition1}, if $\mathcal{C}_A\supseteq K_{1,m-1}(x^*)$, the resulting distributions can be decomposed as
\begin{equation}
    \begin{aligned}
    \label{eq:main12}
        p^{\text{obs}}(ab|xy) = \sum_{e\in[d]}p(e)p^e(ab|xy), ~~~\forall a,b,x,y,
    \end{aligned}
\end{equation}
where $p^e(ab|xy) \in Q$. For each $e$, the behavior $\{p^e(ab|xy)\}$ deterministically produces $a = e$ whenever $x = x^*$. Thus, $p^{\text{obs}}(ab|xy)$ lies within a partially deterministic polytope~\cite{PhDthesis_2014_Eirk}. In Proposition~\ref{proposition 2} in Appendix~\ref{appB}, we prove that no behavior inside such polytope can violate any chain inequality. 

The intuition behind this proof arises from the presence of the chain factor in Eq.~\eqref{eq:chain_manuscript}. Consider a deterministic behavior such that each measurement has a fixed outcome:
\begin{equation}
    \begin{aligned}
    \label{eq:chainconstant}
        A_{i} - B_{i} \equiv q_{2i-1},\\
        B_{i} - A_{i+1} \equiv q_{2i},
    \end{aligned}
\end{equation}
where $q_i\in[d]$ for $i\in\left\{1,\cdots,m\right\}$. The summation of these equations indicates that the classical bound is $I_{d,m}^C = \max_{q_1,\cdots,q_{2m}} \sum_{i = 1}^{2m} \alpha_{q_i}$ under the constraint of~\cite{arXiv_2022_mengyaohu,prl_2017_satwap}
\begin{equation}
    \begin{aligned}
    \label{eq:constraint}
        \sum_{i = 1}^{2m}q_{i} = -1.
    \end{aligned}
\end{equation}
When Alice and Bob share a strong correlation such that the difference between their outcomes are constants, they can also provide $I_{d,m} = \sum_{i = 1} ^{2m} \alpha_{q_{i}'}$ but without the constraint Eq.~\eqref{eq:constraint}, where $q_i' \in [d]$. For instance, this can be achieved by a no-signaling behavior with $m$-input and $d$-output~\cite{prl_2017_satwap}:
\begin{equation}
    \begin{aligned}
    \label{eq:ns_example}
        p(ab|yy) = p(ab|y+1,y) = 
        \begin{cases}
            1/d \quad &\text{if} ~ a = b,\\
            0 \quad &\text{if} ~ a \neq b,
        \end{cases}
    \end{aligned}
\end{equation}
where $y\in\left\{1,\cdots,m\right\}$. For the special case of $(x,y) = (1,m)$, $p(ab|1m) = 1/d$ if $a = b - 1$, otherwise $p(ab|1m) = 0$. For all the other input combinations, we have $p(ab|xy) = 1/d^2$. This behavior also satisfies Eq.~\eqref{eq:chainconstant} but gives $\{q_i' = 0\}_{i=1}^{2m}$. 
Other values of $\{q_i'\}_{i=1}^{2m}$ can be given by considering similar no-signaling behaviors with the same marginal probabilities of $p_A(a|x) = p_B(b|y) = 1/d$. 
Next, if we restrict $A_{x^*}$ to a constant $e\in[d]$, and the other measurements still yield outcomes as before, i.e., $p(ab|x^*y) = \delta_{a,e}p_B(b|y)$, $\forall a, b,y$. For the sake of clarity, we take $x^* = 1$ such that the constants $\{q_i'\}_{i=2}^{2m-1}$ remain unchanged. The resulting no-signaling behavior gives
\begin{equation}
\label{eq:main16}
    I_{d,m} = \sum_{i = 2} ^{2m-1} \alpha_{q_{i}'}  +  {1\over d} \sum_{q'_{1}\in[d]} \alpha_{q'_{1}} + {1\over d} \sum_{q'_{2m}\in[d]} \alpha_{q'_{2m}},
\end{equation}
which is a convex combination of $\sum_{i = 1} ^{2m} \alpha_{q_{i}^{''}}$ with $\{q_i^{''}\}_{i=1}^{2m}$ satisfying the constraint Eq.~\eqref{eq:constraint} and thus it cannot exceed the classical bound $I_{d,m}^C$.

\begin{result}\label{result 1}
\normalfont A bipartite state $\rho_{AB}$ can be used to certify steering-based randomness, determined by whether it can exhibit steerability with restriction to two measurement settings.
\end{result}
\begin{proof}
Suppose there are two measurements performed by Alice such that $\rho_{AB}$ can express steerability, i.e., the resulting assemblage does not admit LHS model, the corresponding assemblage must lie outside the set $\mathcal{R}^{\rm S}_{x^*}$, then the steering-based randomness can be certified. 
Conversely, if $\rho_{AB}$ cannot exhibit steerability when Alice performs any pair measurements, for the assemblage obtained in scenarios where $m \geq 3$, any ensemble $\{\sigma_{a|x}^{\rm obs}\}_{a,x\in\{x_1,x_2\}}$ associated with any two measurements $\{x_1,x_2\}$ must be unsteerable. Based on Proposition~\ref{proposition 1}, such a partially unsteerable assemblage can be given by a 
compatibility structure containing the star subgraph $K_{1,m-1}(x^*)$ for any $x^*$, and thus fails to certify randomness.
\end{proof}

As an example, consider a family of one-way steerable two-qubit states with parameters of $p \in(1/2,1/\sqrt{2}]$ and $\theta \in (0,\pi/4]$:
\begin{equation}
    \begin{aligned}
        \rho_{AB}^{p,\theta} = p|\Psi_\theta\rangle \langle \Psi_\theta| + (1-p)I_A/2\otimes \rho_B^\theta,
    \end{aligned}
\end{equation}
where $|\Psi_\theta\rangle = \cos(\theta)|00\rangle + \sin(\theta)|11\rangle$ and $\rho_B^\theta = \text{Tr}_A\left(|\Psi_\theta\rangle \langle \Psi_\theta|\right)$. In the experiment~\cite{PRL_2017_jinshixu}, these states are steerable with $m\geq 3$ measurements, but cannot demonstrate steerability under any two projective measurements~\cite{np_2010_howard,pra_2016_nicolas}. According to Result~\ref{result 1}, these states fail to certify randomness for arbitrary projective measurements.

So far, we have presented our main findings, which delineate the necessary and sufficient conditions for randomness certification in different scenarios, along with the methods to distinguish useful quantum resources. These results address the scenario where randomness is extracted from a single measurement performed by each user. To improve the quantification of randomness, we consider that randomness can be extracted from multiple measurements. In this framework, we show that the aforementioned conditions, corresponding to the star graph, can be relaxed.

Now Eve optimizes her strategy for a set of measurements $X = \left\{x_1,x_2,\cdots,x_s\right\}$, where $s\leq m$. The target function in Eq.~\eqref{eq:definition_steer_rand} is changed to
\begin{equation}
    \begin{aligned}
    \label{eq:generalization}
        P^{\rm S}_{\text {guess }} \left(X\right) & = \max_{\left\{\sigma_{a \mid x}^{\gamma}\right\}} \sum_{x\in X} p(x) \sum_{\gamma}\operatorname{Tr}\left(\sigma_{a=\gamma(x) \mid x}^\gamma\right),
    \end{aligned}
\end{equation}
where $p(x)>0$ denotes the relative frequencies of inputs, and $\sum_{x\in X}p(x) = 1$. Here Eve's guess for different measurements is a string $\gamma = (a_{x=x_1},\cdots,a_{x = x_s})$, where $\gamma(\cdot)$ being a function from $X$ to $[d]$. Notably, higher randomness can be certified when assuming Eve's strategies for all choices of $x\in X$ are identical.

\begin{result}\label{result 4}
\normalfont Given an assemblage $\{\sigma_{a|x}^{\rm obs}\}$, the guessing probability satisfies $P^{\rm S}_{\text{guess}} \left(X\right) < 1$ iff it lies outside the set $\mathcal{R}_{X}^{\rm S}\coloneqq\left\{\sigma_{a|x}: \sum_a\sigma_{a|x} = \sum_a \sigma_{a|x'},~\forall x,x',~~\mathcal{C}_{A} \supseteq K_{s,m-s} (X)\right\}$ for any $X\subseteq\{1,\cdots,m\}$. 
\end{result}
Here, the structure $K_{s,m-s} (X)$ implies the set of measurements $X$, when combined with every other measurement~$x\notin X$, are compatible. An example is shown in Fig.~\ref{fig:structure}~(c). Since the set $\mathcal{R}_{X}^{\rm S}$ constrained by the structure $K_{s,m-s}(X)$ is a subset of $\mathcal{R}_{x^*}^{\rm S}$ in Proposition~\ref{proposition 1} when $x^*\in X$, an assemblage is more likely to fall outside the former set, thus the requirement of measurement incompatibility for randomness certification is relaxed.

\begin{proof}
The ``only if'' direction is straightforward, if $\mathcal{C}_A\supseteq K_{s,m-s}(X)$, then the set of measurements $X$ can be simulated by a parent POVM. By treating this set as one measurement with $d^s$ outcomes, then $P_{\text{guess}}^{\rm S}(X) = 1 $ according to Proposition~\ref{proposition 1}. For the ``if'' direction, if Eve's guess $\gamma$ is always the same as the outcomes of measurement in $X$, similar to Eq.~\eqref{eq:proposition1}, one can consider these measurements as being carried out by Eve. This allows the observed assemblage to be explained by a scenario where the compatibility structure of the steering party's measurements is $K_{s,m-s}(X)$. In fact, when $s = m$, this result suggests that all steerable states can certify $P_{\text{guess}}^{\rm S} (X)<1$ in Eq.~\eqref{eq:generalization}. 
More details can be found in Appendix~C. 
\end{proof}

\textit{Conclusion.---} In conclusion, we provide the necessary and sufficient conditions for certifying randomness for different scenarios in terms of measurement incompatibility, depicting by a star graph. As a consequence, we prove that violating any chain inequality contributes to certifying randomness. Additionally, we identify quantum states useful for steering-based randomness. 
Finally, we consider the case that takes the input distributions into account, revealing the potential for enhanced randomness certification. 
Our investigation underscores the role of measurement incompatibility in randomness certification, offering insights into uncovering their relation. 

There are several pertinent issues that need to be addressed of attaining a comprehensive understanding of the relationship between entanglement and randomness. For instance, when an assemblage can certify nonzero randomness, the amount of certifiable randomness is related to the maximal amount of the ``$K_{1,m} (x^*)$-partially-unsteerable'' assemblage it contains~(see Appendix~\ref{appD}). 
Moreover, techniques such as testing or quantifying measurement incompatibility in a device-independent or one-sided device-independent manner~\cite{prl_2019_quintino,prr_2021_shiliangchen,prl_2009_wolf,prl_2016_shinliangchen} can be applied to identify behaviors that lead to randomness certification. In other words, the task of randomness certification can be used to witness some specific structures of measurement incompatibility.


\begin{acknowledgments}
We acknowledge enlightening discussions with Otfried G\"uhne, H. Chau Nguyen and Mengyao Hu. This work was supported by Beijing Natural Science Foundation (Grant No.~Z240007), and the National Natural Science Foundation of China (No.~12125402, No~12350006, No~123B2073), the Innovation Program for Quantum Science and Technology (No.~2021ZD0301500). Y.L. acknowledges financial support from the China Scholarship Council (Grant No.~202306010353).
J.T. acknowledges the support received by the Dutch National Growth Fund (NGF), as part of the Quantum Delta NL programme, as well as the support received from the European Union’s Horizon Europe research and innovation programme through the ERC StG FINE-TEA-SQUAD (Grant No. 101040729). The views and opinions expressed here are solely those of the authors and do not necessarily reflect those of the funding institutions. Neither of the funding institution can be held responsible for them.
\end{acknowledgments}

\bibliography{Ref}

\appendix

\onecolumngrid

\newpage

\begin{appendix}

\section{Proof of Proposition 1}
\label{appA}
In Proposition~\ref{proposition 1} in manuscript, we demonstrate that given an assemblage $\{\sigma_{a|x}^{\rm obs}\}$, the guessing probability satisfies $P_g^{\rm S}(x^*) < 1$ iff it lies outside the set $\mathcal{R}^{\rm S}_{x^*}$ for any $x^*\in\{1,\cdots,m\}$, where $\mathcal{R}^{\rm S}_{x^*}$ is defined as:
\begin{equation}
    \left\{ \sigma_{a|x}: 
    \begin{array}{ll}
        &\sigma_{a|x} = {\rm Tr}_A\left[\left(M_{a|x}\otimes \mathbb{I}^B \right)\rho_{AB}\right], \mathcal{C}_{A} \supseteq K_{1,m-1} (x^*),\\
        &M_{a|x}\succeq 0, \sum_a M_{a|x} = \mathbb{I}^A, \rho_{AB}\succeq 0, {\rm Tr} (\rho_{AB}) = 1, 
    \end{array} \right\}.
\end{equation}
In the following proof, we will prove the ``only if'' direction. The ``if'' direction has been proved in the work~\cite{PRL_2024_yi}, we will also briefly review it in the operational interpretation.
\begin{proof}
If $\mathcal{C}_A \supseteq K_{1,m-1}\left(x^*\right)$, then there exist $m-1$ parent POVMs $\{G_{\lambda}^i\}_{\lambda}$ for $i=1,2,\cdots,m-1$, such that Alice's measurements can be decomposed as 
\begin{equation}
    \begin{aligned}
        M_{a|x} &= \sum_{\lambda} p^i(a|x,\lambda)G^i_\lambda, \quad \forall a,x \in e_i, \jt{\quad} e_i\in E\left(K_{1,m-1}(x^*)\right),
    \end{aligned}
\end{equation}
where $e_i$ is an edge of the star graph with $i\in\{1,\cdots,m-1\}$ and $E\left(K_{1,m-1}(x^*)\right)$ denotes the set of edges of graph $K_{1,m-1}(x^*)$. For the sake of simplicity, we will directly use $e_i$ to denote the edge in the star graph $K_{1,m-1}(x^*)$. Here we note that the notation $x\in e_i$ means the measurement $x$ belongs to the subset of measurements associated to the edge $e_i$. Since $p^i(a|x,\lambda) = \sum_\mu D_{\text{Local}}(a|x,\mu)p^i(\mu|\lambda)$, one can rewrite Alice's measurements as $M_{a|x}= \sum_{\mu}D_{\text{Local}}(a|x,\mu)\mathcal{G}^i_\mu$, where $\mathcal{G}^i_\mu = \sum_{\lambda}p^i(\mu|\lambda)G_\lambda^i$ and $\sum_\mu \mathcal{G}^i_\mu = \mathbb{I}^A$. Here the deterministic behaviors are $D_{\text{Local}}(a|x,\mu) = \delta_{a,\mu(x)}$ with $\mu(\cdot)$ being a function from $\{1, \cdots, m\}$ to $[d]$, where $[d]\coloneqq \{0,1,\cdots,d-1\}$. 
We note that the number of deterministic behaviors are the same for all edges $e_i$ in this work, thereby there is no subscript $i$ in the hidden variables. It can be simply generalized to the scenario that different measurements $x$ produce different numbers of outcomes. 

Thus, for a given bipartite state $\rho_{AB}$, the observed assemblage can be decomposed as
\begin{equation}
    \begin{aligned}
    \label{eq:partiallyunsteerable}
    \sigma_{a|x}^{\text{obs}} &= \sum_{\mu}D_{\text{Local}}(a|x,\mu)\tau_\mu^{i},\quad \forall a,x \in e_i,
    \end{aligned}
\end{equation}
where $\tau_\mu^i \coloneqq \text{Tr}_A\left[\left(\mathcal{G}^i_\mu \otimes \mathbb{I}^B\right)\rho_{AB}\right]$. We denote the assemblages with the decomposition Eq.~\eqref{eq:partiallyunsteerable} as the ``$K_{1,m-1} (x^*)$-partially-unsteerable'' assemblage. Then, we can directly give an optimal solution of Eq.~\eqref{eq:definition_steer_rand} in the manuscript. For clarity, we denote it as
\begin{equation}
\label{eq:A4}
    \xi_{a|x}^{e} \coloneqq \sum_{\mu\in\Lambda^{(e)}}D_{\text{Local}}(a|x,\mu) \tau_\mu^i, \quad \forall e, a, x \in e_i,
\end{equation}
where $\Lambda^{(e)} \coloneqq \{\mu|\mu(x^*) = e\}$. It is a subset of all deterministic behaviors in which the behaviors always produce outcome $e$ when receiving input $x^*$. We note that the constructions of $x = x^*$ for all edges $e_i$ are the same:
\begin{equation}
\begin{aligned}
    \xi_{a=e|x^*}^e &= \sum_{\mu}D_{\text{Local}}(a = e|x^*,\mu) \tau_\mu^i = \sigma_{e|x^*}^{\text{obs}},\quad \forall e,i.
\end{aligned}
\end{equation}
For $a\neq e$, we have $\xi_{a|x^*}^e = 0$. With this construction, Eve can always guess the outcomes of measurement $x^*$ correctly:
\begin{equation}
    \begin{aligned}
        \sum_e\text{Tr}\left(\xi_{a=e|x^*}^e\right) = \sum_\mu \text{Tr}\left(\tau_\mu^i\right) = 1.
    \end{aligned}
\end{equation}
Finally, one can easily check that this construction aligns with the observed assemblage, i.e., $\sum_e \xi_{a|x}^e = \sigma_{a|x}^{\rm obs}$. Since all the vertices in the star graph are connected with the vertex $x^*$, it also satisfies the no-signaling conditions
\begin{equation}
    \begin{aligned}
        \sum_a \xi_{a|x^*}^e = \sum_a\xi_{a|x^i}^{e} = \sum_{\mu\in\Lambda^{(e)}} \tau_{\mu}^i, \quad \forall e,(x^*,x^i) =  e_i.
    \end{aligned}
\end{equation}
Therefore, when $\mathcal{C}_A \supseteq K_{1,m-1}\left(x^*\right)$, the observed assemblage can be decomposed by $\{\xi_{a|x}^e\}$ and provides $P_{g}^{\rm S}(x^*) = 1$, which concludes the proof.
\end{proof}

In the operational interpretation, for the solution $\{\xi_{a|x}^e\}$ giving $P_g^{\rm S}(x^*) = 1$ of Eq.~\eqref{eq:definition_steer_rand} in the manuscript, since every bipartite no-signaling assemblage admits quantum realization~\cite{pr_1957_ghjw,hpa_1989_ghjw,pla_1993_ghjw,mpcps_1936_schro}, one can find a physical scenario for the constructed assemblage $\{\xi_{a|x}^e\}$. For each $e$, one can always find a quantum state $\varrho_{AB}^e$ and POVMs $\left\{\mathcal{M}_{a|x}^e\right\}_{a,x}$ such that
\begin{equation}
    \begin{aligned}
        \label{eq:quantumrealization}
        \xi_{a|x}^e = p(e)\text{Tr}_A\left[ \left(\mathcal{M}^e_{a|x}\otimes\mathbb{I}^B \right)\varrho_{AB}^e\right],\quad \forall e, a, x \in e_i,
    \end{aligned}
\end{equation}
where $p(e) = \sum_a\text{Tr}\left(\xi_{a|x}^e\right) = \text{Tr}\left(\sigma_{e|x^*}^{\text{obs}}\right)$. 
This construction describes a scenario in which, when Eve gives a guess $e$ with probability $p(e)$, Alice performs measurements $\left\{\mathcal{M}_{a|x}^e\right\}_{a,x}$ on the state $\varrho_{AB}^e$ according to $e$. 
It can be done by distributing a state among Alice, Bob, and Eve, and considering an auxiliary classical `flag' system~\cite{njp_2015_paul} labeled by $A'$:
\begin{equation}
    \begin{aligned}
    \label{eq:operational}
        \rho_{EAA'B} &= \sum_e p(e)|ee\rangle_{EA'}\langle ee| \otimes \varrho_{AB}^e,\\
        \xi_{a|x}^e &= \text{Tr}_{EAA'} \left[\left( |e\rangle_E\langle e|\otimes \mathcal{M}_{a|x}\otimes \mathbb{I}^B\right) \rho_{EAA'B} \right], \quad \forall a,e,x,
    \end{aligned}
\end{equation}
where Eve performs measurement $\left\{|e\rangle_E\langle e|\right\}_e$ and Alice performs measurements $\left\{\mathcal{M}_{a|x} = \sum_e|e\rangle_{A'}\langle e| \otimes \mathcal{M}_{a|x}^e\right\}_{a,x}$. Here $\left\{|e\rangle\right\}_{e}$ are orthogonal states for $e \in [d]$, and the role of `flag' system $A'$ is to remove the superscript $e$ in $\left\{\mathcal{M}_{a|x}^e\right\}_{a,x}$. When Eve produces outcome $e$ with probability $p(e)$, the conditional state held by Alice and Bob is
\begin{equation}
    \begin{aligned}
        \rho_{A'AB}^e = |e\rangle_{A'}\langle e|\otimes \varrho_{AB}^e.
    \end{aligned}
\end{equation}
Then the additional degree of freedom $|e\rangle_{A'}\langle e|$ is read by Alice, such that she can perform measurements $\left\{\mathcal{M}_{a|x}^{e}\right\}_{a,x}$ according to the guess $e$. 
For the sake of clarity, we treat the subsystems $A$ and $A'$ together as $\boldsymbol{A}$. 
Based on Eq.~\eqref{eq:operational}, the observed assemblage can be rewritten as
\begin{equation}
    \begin{aligned}
        \sigma_{a|x}^{\text{obs}} = \sum_e \xi_{a|x}^e = \text{Tr}_{E\boldsymbol{A}}\left[ \left( \mathcal{M}_{a|x}^{E\boldsymbol{A}} \otimes \mathbb{I}^B\right)\rho_{E\boldsymbol{A}B} \right], \quad \forall a,x,
    \end{aligned}
\end{equation}
where $\rho_{E\boldsymbol{A}B} = \sum_ep(e)|e\rangle_{E}\langle e| \otimes \rho_{\boldsymbol{A}B}^e$ and $\mathcal{M}_{a|x}^{E\boldsymbol{A}} = \mathbb{I}^E\otimes \mathcal{M}_{a|x}$ for $x\neq x^*$. For $x = x^*$, since $P_g^{\rm S}(x^*) = 1$, we have ${\rm Tr} \left(\xi_{a|x^*}^e \right) = 0$ whenever $a \neq e$. Then, $\sum_{e}\xi_{a|x^*}^e = \sum_{a'}\xi_{a'|x^*}^{e = a}$ and $\mathcal{M}_{a|x^*}^{E\boldsymbol{A}} = |a\rangle_E\langle a|\otimes \mathbb{I}^{\boldsymbol{A}}$. Additionally, this result also implies that a ``$K_{1,m-1} (x^*)$-partially-unsteerable'' assemblage (Eq.~\eqref{eq:partiallyunsteerable}) can be explained by a scenario where $\mathcal{C}_A \supseteq K_{1,m-1}(x^*)$. In other words, the ``$K_{1,m-1} (x^*)$-partially-unsteerable'' set $\left\{\sigma_{a|x}: \sum_a\sigma_{a|x} = \sum_a\sigma_{a|x'},~\forall x,x',~ \sigma_{a|x} = \sum_{\mu} D_{\rm Local}(a|x,\mu)\tau_{\mu}^i,~\forall a,x\in e_i\right\}$ and the set $\left\{\sigma_{a|x}: \sum_a\sigma_{a|x} = \sum_a\sigma_{a|x'},~\forall x,x',~ \mathcal{C}_A \supseteq  K_{1,m-1}(x^*)\right\}$ are the same. 

\section{Proof of Result 2}
\label{appB}
In the manuscript, we claim that all the chain inequalities with $m\geq 2$ inputs and $d\geq 2$ outputs in the form of 
\begin{equation}
    \begin{aligned}
    \label{eq:chaininequality}
        I_{d,m} \coloneqq \sum_{k=0}^{d-1}\alpha_k \sum_{i=1}^m \left[p(A_i-B_i = k) + p(B_i-A_{i+1}=k)\right]\leq I_{d,m}^{C}
    \end{aligned}
\end{equation}
cannot be violated when the compatibility structure $\mathcal{C}_A$ of Alice's measurements is isomorphic to a hypergraph containing $K_{1,m-1}\left(x^*\right)$, i.e., $K_{1,m-1}\left(x^*\right) \subseteq \mathcal{C}_A$. For the case of $m = 2$, $K_{1,m-1}\left(x^*\right) \subseteq \mathcal{C}_A$ means Alice performs compatible measurements, so the chain inequalities cannot be violated. In the following, we focus on the cases of $m\geq 3$.


Before proving this result, we first consider the joint probability distributions given when Alice (Bob) receives $m_A$ ($m_B$) inputs and produces $d_A$ ($d_B$) outputs, and treat this behavior as a vector $\boldsymbol{p} = \left(p(00|11),\cdots,p(d_A-1,d_B-1|m_Am_B)\right) ^{T}\in \mathbb{R}^{t}$ with $t = m_Am_Bd_Ad_B$. 
Therefore, two convex subsets \textit{embedded} in $t$-dimensional subspace of $\mathbb{R}^t$ are characterized as:
\begin{align}
\mathcal{P}(d_A,d_B,m_A,m_B) &\coloneqq \left\{ \boldsymbol{p} \in \mathbb{R}^{t} : p(ab|xy) \geq 0, \sum_{ab} p(ab|xy) = 1, \forall a, b, x, y \right\},   \nonumber  \\ 
\mathcal{NS}(d_A,d_B,m_A,m_B) &\coloneqq \left\{ \boldsymbol{p} \in \mathbb{R}^{t} : p(ab|xy) \geq 0, \sum_{ab} p(ab|xy) = 1, \forall a,b,x,y, ~\sum_{a} p(ab|xy) = \sum_{a} p(ab|x'y), ~\forall b, x, x', y,\right. \\ \nonumber
&\qquad \qquad \left.~\sum_{b} p(ab|xy) = \sum_{b} p(ab|xy'),  \forall a, x, y, y'\right\}.
\end{align}
A subset of behaviors $\boldsymbol{p} \in \mathbb{R}^{t}$ can be viewed as points belonging to the probability space $\mathcal{P}$ of dimension $\dim \left(\mathcal{P}\right) = \left( d_Ad_B-1 \right)m_Am_B$. When assuming that the local marginal probabilities of Alice are independent of Bob's measurement setting $y$ (and the other way around), the behaviors lie within the no-signaling subspace $\mathcal{NS}$ of dimension $\dim \left(\mathcal{NS}\right) = m_A(d_A-1) + m_B(d_B - 1) + m_Am_B(d_A-1)(d_B-1)$.

For the chain inequality, as mentioned in the manuscript, its classical bound is given by assigning a deterministic outcome to each of the $A_i$ and $B_i$ such that~\cite{prl_1982_fine,arXiv_2022_mengyaohu,prl_2017_satwap}
\begin{equation}
    \begin{aligned}
    \label{eq:chaincharacter}
        A_1 - B_1 &= q_1, \\
        B_1 - A_2 &= q_2, \\
        &\vdots\\
        B_m - (A_1 + 1) &= q_{2m},
    \end{aligned}
\end{equation}
where $q_i\in[d]$ for $i=1,\cdots,2m$, and $\sum_{i=1}^{2m} q_i = -1$. Here we note that $A_{m+1} \equiv A_1 + 1$. Therefore, the classical bound is $I_{d,m}^C = \max_{q_1,\cdots, q_{2m}} \left(\sum_{i=1}^{2m}\alpha_{q_i} \right)$. 
Then, we will first define a set of probability distributions that cannot violate the chain inequality.

\begin{definition}
\label{definition1}
    Let $\boldsymbol  q\coloneqq (q_1,q_2,\cdots,q_{2m-1})^T \in \mathbb{Z}_d^{2m-1}$ be a vector, a set of probability distributions $ \mathbb{P}_{\boldsymbol{q}}^c \subset \mathcal{P}(d,d,m,m)$ is defined in which the behaviors $\boldsymbol{p} \in \mathcal{P}(d,d,m,m)$ satisfy
    \begin{equation}
        \label{eq:chainconstraintP}
        p\left(A_i - B_i = q_{2i-1}\right) = p\left(B_i - A_{i+1} = q_{2i}\right) = 1,\qquad \forall i\in\{1,2,\cdots,m\},
    \end{equation}
    where $q_{2m} \coloneqq -\|\boldsymbol  q\|_1 - 1$ and $\|\cdot\|_1$ denotes the $\ell^1$-norm.
\end{definition}

\begin{lemma}\label{lemma addition2}
For any vector $\boldsymbol q \in \mathbb{Z}_d^{2m-1}$, no behavior in the set $\mathbb{P}_{\boldsymbol{q}}^c$ can violate any chain inequality Eq.~\eqref{eq:chaininequality}. Thus, no behavior in ${\rm Conv}\left( \cup_{\boldsymbol{q}} \mathbb{P}_{\boldsymbol{q}}^c\right)$ can violate any chain inequality, where ${\rm Conv}\left( \cdot \right)$ means the convex hull of a given set. 
\end{lemma}
\begin{proof}
For any vector $\boldsymbol{q}\in \mathbb{Z}_d^{2m-1}$, a behavior satisfying Eq.~\eqref{eq:chainconstraintP} means that the difference between the outcomes of measurements $A_i$ ($B_i$) and $B_i$ ($A_{i+1}$) is constant $q_{2i-1}$ ($q_{2i}$), although their outcomes do not necessarily need to be produced deterministically. Since the summation of all these constants is $\sum_{i=1}^{2m}q_i = -1$, the behaviors in the set $\mathbb{P}_{\boldsymbol{q}}^c$ as well as any behaviors lying within their convex hull cannot violate any chain inequality. 
\end{proof}
\begin{definition}
    For a given input subset $\mathcal{I}\subseteq\{A_1,\cdots,A_m\}\cup\{B_1,\cdots,B_m\}$, we define a projection operation $\mathrm{\Pi}_{\mathcal{I}} $ such that for a behavior $\boldsymbol{p} \in \mathcal{P}(d,d,m,m)$, $\mathrm{\Pi}_{\mathcal{I}} \boldsymbol{p}$ is defined by the elements
    \begin{equation}
        \mathrm{\Pi}_{\mathcal{I}} {p}(ab|xy) = p(ab|xy),  \qquad A_x,B_y \in \mathcal{I}.
    \end{equation}
    We also define a subset of input pairs considered in this chain inequality as $\mathcal{I}^c = \left\{(A_y,B_y),(A_{y+1},B_y) : y = 1,\cdots,m\right\}$, and $\mathrm{\mathrm{\Pi}}_{\mathcal{I}^{c}} \boldsymbol{p} = \left(p(00|11),\cdots,p(d-1,d-1|1m)\right)^T \in \mathbb{R}^{2md^2}$, where $A_{m+1} \equiv A_1$.
\end{definition}
Lemma~\ref{lemma addition2} implies that a behavior $\boldsymbol{p} \in \mathcal{NS}(d,d,m,m)$ cannot violate any chain inequality if the vector $\mathrm{\Pi}_{\mathcal{I}^c}\boldsymbol{p}$ can be written as a convex combination of $\{\mathrm{\Pi}_{\mathcal{I}^{c}}\boldsymbol{P}_{\boldsymbol  q} \}_{\boldsymbol{q}}$, where $\boldsymbol{P}_{\boldsymbol  q}\in\mathbb{P}_{\boldsymbol{q}}^c$ and $\boldsymbol{q}\in \mathbb{Z}_d^{2m-1}$. 

In the following, to prove Result~\ref{result 3} in the manuscript, for the sake of concreteness, let us set $x^* = 1$. The cases for other $x^*\neq 1$ can be analyzed in an analogous manner, as relabelling the inputs does not change the chain factor. 

In Lemma~\ref{lemma addition}, we will focus on the input subset of $\mathcal{I}^c\backslash \left(A_{x^*},\cdot\right)$, which means we remove the input pairs that contain $A_{x^*}$ in $\mathcal{I}^c$. Thus, we consider the no-signaling behaviors $\boldsymbol{\kappa}\in\mathcal{NS}(d,d,m-1,m)$ in which Alice only receives $m-1$ inputs (remove $x^*$). For a given vector $\boldsymbol{l}\coloneqq \left(q_2,q_3,\cdots,q_{2m-1}\right) \in \mathbb{Z}^{2m-2}_d$ with elements $\{q_i\}$ related to Eq.~\eqref{eq:chainconstraintP}, we define an auxiliary function $\mathcal{F}_{\boldsymbol{l}}^m(t)$ based on $\boldsymbol{\kappa}$ and prove the following properties of it. These properties will be useful in Lemma~\ref{lemma 2}, which shows that the behavior $\boldsymbol{\kappa}$ can be written as a convex combination of a set of behaviors $\{\boldsymbol{F}_{\boldsymbol{l}}\}_{\boldsymbol{l}}$ satisfying Eq.~\eqref{eq:chainconstraintP}. For arbitrary $\boldsymbol{l}\in\mathbb{Z}_{d}^{2m-2}$, the behavior $\boldsymbol{F}_{\boldsymbol{l}}$ means that the difference between outcomes of measurements $(A_x,B_y) \in \mathcal{I}^c\backslash(A_{x^*},\cdot)$ are constants, such that the behavior $\boldsymbol{F}_{\boldsymbol{l}}$ gives values of some terms $\sum_{i=2}^{2m-1}\alpha_{q_i}$ in the chain Bell inequality Eq.~\eqref{eq:chaininequality}, like Eq.~\eqref{eq:main16} in the manuscript, which is the key technical requirement in the proof of Result~\ref{result 3}.

Then, in Proposition~\ref{proposition 2}, we show that any no-signaling behavior $\boldsymbol{p}\in\mathcal{NS}(d,d,m,m)$ that deterministically produces outcome when $x = x^*$, e.g., $\{p^e(ab|xy)\}$ in Eq.~\eqref{eq:main12} in the manuscript, is fully determined by its projection $\mathrm{\Pi}_{\{A_1,\cdots,A_m\}\cup\{B_1,\cdots,B_m\}\backslash A_{x^*}} \boldsymbol{p} \eqqcolon \boldsymbol{\kappa}\in\mathcal{NS}(d,d,m-1,m)$. Based on the decomposition of $\boldsymbol{\kappa}$ found in Lemma~\ref{lemma 2}, we will prove that the behavior $\boldsymbol{p}$ can always be written as a convex combination of behaviors $\{\boldsymbol{P}_{\boldsymbol{q}}\in\mathbb{P}_{\boldsymbol{q}}^c \}_{\boldsymbol{q}}$, where $\boldsymbol{q} = (q_1,\boldsymbol{l})\in \mathbb{Z}_d^{2m-2}$. Here, a behavior $\boldsymbol{P}_{\boldsymbol{q}}$ gives a value of $\sum_{i=2}^{2m-1}\alpha_{q_{i}} + \alpha_{q_1} + \alpha_{q_{2m}}$ of the chain inequality Eq.~\eqref{eq:chaininequality}, where $\sum_{i = 1}^{2m} q_{i} = -1$. Thus, this partially deterministic behavior $\boldsymbol{p}$ cannot violate any chain inequality.

\begin{lemma} \label{lemma addition}
Let $\boldsymbol  \kappa \in \mathcal{NS} (d,d,m-1,m)$ be a no-signaling behavior, and let $\boldsymbol{l}\coloneqq \left(q_2,q_3,\cdots,q_{2m-1}\right) \in \mathbb{Z}^{2m-2}_d$ be a vector. A function induced by $\boldsymbol  \kappa$ is defined by 
\begin{equation}
    \begin{aligned}
    \mathcal{F}^{m}_{\boldsymbol l}(t) &\coloneqq \kappa_B(t+q_2|1)\prod_{k = 2}^m g_{\boldsymbol  l}^{k}(t - c_k),~~~\text{ with }~~~ g_{\boldsymbol  l}^{k}(t) &\coloneqq \frac{ \kappa(t,t+q_{2k-2}|k,k-1)\kappa(t ,t-q_{2k-1} |kk)}{\kappa_B\left(t + q_{2k-2}| k-1\right)\kappa_{A}\left( t | k\right)},
    \end{aligned}
\end{equation}
where $c_k \coloneqq \sum_{i = 3}^{2k-2}q_i$ with $c_2 \coloneqq 0$, and $g_{\boldsymbol{l}}^{k}(t) \coloneqq 0$ if any marginal probability in the denominator is zero. Here $\kappa_A(a|x) \coloneqq \sum_{b}\kappa(ab|xy)$ and $\kappa_B(b|y) \coloneqq \sum_{a}\kappa(ab|xy)$ are the marginal probabilities. Then, this function satisfies the following properties:
\begin{equation}
    \begin{aligned}
    \label{eq:mathcalF}
    \sum_{q_{2i},\cdots,q_{2m-1}\in[d]} \mathcal{F}^{m}_{\boldsymbol  l}(t) &= \mathcal{F}^{i}_{\boldsymbol  l}(t),~~\forall \boldsymbol  l, t, i\geq 2,~~~\sum_{q_{2},\cdots,q_{2i-3}\in[d]} \mathcal{F}^{i}_{\boldsymbol  l}(t + c_{i}) &= \kappa_B(t+q_{2i-2}|i-1) g_{\boldsymbol{l}}^i(t),~~~\forall \boldsymbol  l, t, i\geq 3.
    \end{aligned}
\end{equation}
\end{lemma}

\begin{proof}
Since every probability comprising the marginal must be zero if the marginal probability is zero, so we define $g_{\boldsymbol{l}}^{k}(t) \coloneqq 0$ if any marginal probability in the denominator is zero. 
Then, we will first prove the first property in Eq.~\eqref{eq:mathcalF}. Since $\sum_a\kappa_A(a|m) = \sum_b\kappa_B(b|m-1) = 1$, at least some values of $q_{2m-2},q_{2m-3}\in[d]$ such that the marginal probabilities in denominator of $g_{\boldsymbol{l}}^m(t - c_m)$, i.e., $\kappa_A(t-c_m| m)$ and $\kappa_B(t-c_m+q_{2m-2}| m-1)$, are larger than zero for any $t$ and $\boldsymbol l\backslash (q_{2m-3},q_{2m-2})$. We denote the sets of these values $q_{i}$ as $q_{i,t,\boldsymbol{l}\backslash (q_{2m-3},q_{2m-2})}^{>0}\subseteq [d]$ with $i = 2m-2, 2m-3$. For the sake of simplicity, we just write $q_{i}^{>0}$, and then we have
\begin{equation}
    \begin{aligned}
        \sum_{q_{2i},\cdots,q_{2m-1}\in[d]}\mathcal{F}_{\boldsymbol{l}}^m(t) &= \sum_{q_{2i},\cdots,q_{2m-3}\in[d]} \mathcal{F}_{\boldsymbol{l}}^{m-1}(t)\sum_{q_{2m-2}\in q_{2m-2}^{>0}}
        \frac{\kappa(t - c_m,t - c_m +q_{2m-2} |m,m-1)}{\kappa_B\left(t - c_m + q_{2m-2}| m-1\right)} \\ &= \sum_{q_{2i},\cdots,q_{2m-4}\in[d]} \sum_{q_{2m-3}\in q_{2m-3}^{>0}}\mathcal{F}_{\boldsymbol{l}}^{m-1}(t) \\
        &= \sum_{q_{2i},\cdots,q_{2m-3}\in[d]}\mathcal{F}_{\boldsymbol{l}}^{m-1}(t) , ~~~\forall t, \boldsymbol{l}, i\geq 2, m\geq 3,
    \end{aligned}
\end{equation}
where the second equality is because the values $q_{2m-2}\notin q_{2m-2}^{>0}$ give that $\kappa_A(t-c_m|m) = 0$, thereby $\kappa(t-c_m,t-c_m +q_{2m-2}|m,m-1) = 0$. So here the sum over $q_{2m-2}\in q^{>0}_{2m-2}$ is the same as summing over $q_{2m-2} \in[d]$. Similarly, the third equality is because $\mathcal{F}_{\boldsymbol{l}}^{m-1}(t) = 0$ for the values $q_{2m-3}\notin q_{2m-3}^{>0}$. Therefore, we have $\sum_{q_{2i},\cdots,q_{2m-1}\in[d]} \mathcal{F}^{m}_{\boldsymbol  l}(t) = \mathcal{F}^{i}_{\boldsymbol l}(t)$, $\forall t, \boldsymbol{l}, i\geq 2$. 

Similarly, for the second property in Eq.~\eqref{eq:mathcalF}, at least some values of $q_{2k-1},q_{2k-2}\in[d]$ such that the denominator of $g_{\boldsymbol{l}}^k(t + c_i - c_k)$ is larger than zero for any $t$ and $\boldsymbol l\backslash (q_{2k-2},q_{2k-1})$, we also denote them as $q_{i}^{>0}$. We have
\begin{equation}
    \begin{aligned}
        \sum_{q_{2},\cdots,q_{2i-3}\in[d]} \mathcal{F}^{i}_{\boldsymbol  l}(t + c_{i}) & = \sum_{q_{2},\cdots,q_{2i-3}\in[d]} \kappa_B(t+c_i + q_2|1) \prod_{k = 2}^{i} g_{\boldsymbol{l}}^k(t+c_i - c_k)\\
        & = \sum_{q_{4},\cdots,q_{2i-3}\in[d]} \left( \sum_{q_3\in q_3^{>0}}\kappa(t+c_i ,t+c_i-q_{3} |22) \right) \prod_{k = 3}^{i}g_{\boldsymbol  l}^{k}(t +c_i - c_k) \\
        & = \sum_{q_{4},\cdots,q_{2i-3}\in[d]} \kappa_B(t+c_i + q_2 -(q_2 + q_{3})|2)\prod_{k = 3}^{i}g_{\boldsymbol  l}^{k}(t +c_i - c_k),~~~\forall t, \boldsymbol{l}, i\geq 3.
    \end{aligned}
\end{equation}
Therefore, we have $\sum_{q_{2},\cdots,q_{2i-3}\in[d]} \mathcal{F}^{i}_{\boldsymbol l}(t + c_{i}) = \kappa_B(t+q_{2i-2}|i-1) g_{\boldsymbol{l}}^i(t)$, which concludes this proof.
\end{proof}

In Lemma~\ref{lemma 2}, we will use the properties proved in Lemma~\ref{lemma addition} and show that the behavior $\boldsymbol{\kappa}\in\mathcal{NS}(d,d,m-1,m)$ can be written as a convex combination of a set of behaviors $\{\boldsymbol{F}_{\boldsymbol{l}}\}_{\boldsymbol{l}}$ satisfying Eq.~\eqref{eq:chainconstraintP}. Here $\{\boldsymbol{F}_{\boldsymbol{l}}\}_{\boldsymbol{l}}$ are defined by $\mathcal{F}_{\boldsymbol{l}}^m(t)$.

\begin{lemma}\label{lemma 2}
Given a no-signaling behavior $\boldsymbol  \kappa \in \mathcal{NS} (d,d,m-1,m)$, and let $\boldsymbol{l}\coloneqq \left(q_2,q_3,\cdots,q_{2m-1}\right) \in \mathbb{Z}^{2m-2}_d$, $\boldsymbol{q} \coloneqq ( q_1 ,\boldsymbol{l}) \in \mathbb{Z}_d^{2m-1}$ be vectors, then there exist some behaviors $\boldsymbol F_{\boldsymbol {l}}\in\mathbb{P}_{\boldsymbol{q}}^c\subset \mathcal{P}(d,d,m,m)$, in which the probabilities associated with the measurement settings $\mathcal{I}^c\backslash \left(A_{x^*},\cdot\right)$ are defined by $\boldsymbol{\kappa}$ based on Lemma~\ref{lemma addition}:
\begin{equation}
    \begin{aligned}
    \label{eq:chain_construct}
        F_{\boldsymbol  l}\left( t,t - q_{2i-1}|ii\right) = F_{\boldsymbol  l}\left( t,t+ q_{2i-2}|i,i-1\right) \coloneqq \mathcal{F}^{m}_{\boldsymbol  l}(t + c_i) /f(\boldsymbol  l) ,~~~\forall \boldsymbol  l, t, i\in \left\{2,\cdots,m\right\},
    \end{aligned}
\end{equation}
where $f(\boldsymbol  l) \coloneqq \sum_{t\in[d]} \mathcal{F}^{m}_{\boldsymbol  l}(t)$ and $\sum_{\boldsymbol  l} f(\boldsymbol  l) = 1$ with $f(\boldsymbol  l)\geq 0$. 
Then, the behavior $\boldsymbol{\kappa}$ can be decomposed as
\begin{equation}
    \begin{aligned}
    \label{eq:kappadecomposition}
        \mathrm{\Pi}_{\mathcal{I}^c\backslash \left(A_{x^*},\cdot\right)}{\boldsymbol{\kappa}} = \mathrm{\Pi}_{\mathcal{I}^c\backslash \left(A_{x^*},\cdot\right)}\sum_{\boldsymbol l}f(\boldsymbol  l) \boldsymbol{F}_{\boldsymbol  l} ,
    \end{aligned}
\end{equation}
where $\sum_{\boldsymbol  l}$ means $\sum_{q_2,\cdots,q_{2m-1}\in [d]}$. 
\end{lemma}

\begin{proof}
Firstly, we note that $f(\boldsymbol l)= 0$ indicates $\mathcal{F}^{m}_{\boldsymbol l} (t) = 0$ for any $t$, then we can simply define the corresponding distribution as a normalized distribution without loss of generality. 
Then, for each $\boldsymbol{l}\in\mathbb{Z}_d^{2m-2},$ we will prove that a behavior $\boldsymbol{F}_{\boldsymbol{l}}$ with probability distributions defined by Eq.~\eqref{eq:chain_construct} belongs to  $\mathbb{P}_{\boldsymbol{q}}^c$. For the input pairs in the set $\mathcal{I}^c\backslash(A_{x^*},\cdot)$, each distribution $\{p(ab|xy)\}_{ab}$ in $\mathbb{P}_{\boldsymbol{q}}^c$ has at most $d$ nonzero probabilities, e.g., $p(a,b \neq a-q_{2i-1}|ii) = 0$, thereby Eq.~\eqref{eq:chain_construct} is enough to describe the distributions associated with the input pairs in $\mathcal{I}^c\backslash(A_{x^*},\cdot)$. 
Thus, $F_{\boldsymbol{l}}(t,b\neq t-q_{2i-1}|ii) \coloneqq 0$, and the distributions in Eq.~\eqref{eq:chain_construct} are normalized: $\sum_{ab}F_{\boldsymbol{l}}(ab|ii) = \sum_{t\in[d]}F_{\boldsymbol{l}}(t,t-q_{2i-1}|ii) = 1$ for any $\boldsymbol{l}$, and likewise for $F_{\boldsymbol{l}}(t,t+q_{2i-2}|i,i-1)$. 
For other input pairs, since $\boldsymbol{F}_{\boldsymbol{l}}$ does not necessarily be no-signaling, one can simply construct normalized probability distributions such that $\boldsymbol{F}_{\boldsymbol{l}}\in \mathbb{P}_{\boldsymbol{q}}^c \subset \mathcal{P}(d,d,m,m)$, and $\boldsymbol{F}_{\boldsymbol{l}}$ has the same projection as in Eq.~\eqref{eq:chain_construct}. 

Additionally, Based on Lemma~\ref{lemma addition}, we have $\mathcal{F}_{\boldsymbol{l}}^m(t),f(\boldsymbol{l})\geq 0$, and 
\begin{equation}
    \begin{aligned}
    \label{eq:contraction}
       \sum_{\boldsymbol  l} f(\boldsymbol  l) = \sum_{t, \boldsymbol  l} \mathcal{F}^{m}_{\boldsymbol  l}(t) = \sum_{t,q_2,q_3\in[d]} \mathcal{F}_{\boldsymbol  l}^2(t) = \sum_{t,q_2,q_3\in[d]} \frac{\kappa(t ,t+q_{2} |21)\kappa(t ,t-q_{3}|22)}{\kappa_A\left(t\mid2\right) } = 1.
    \end{aligned}
\end{equation}

Then, we will show that the vector $\mathrm{\Pi}_{\mathcal{I}^{c}\backslash (A_{x^*},\cdot)}\boldsymbol{\kappa}$ can be written as a convex combination of the set of vectors $\{ \mathrm{\Pi}_{\mathcal{I}^{c}\backslash (A_{x^*},\cdot)} \boldsymbol  F_{\boldsymbol  l} \}_{\boldsymbol{l}}$. For the input pair $(A_{i},B_{i-1})$ and a given $q'_{2i-2}\in[d]$, since $\boldsymbol{F}_{\boldsymbol{l}}\in\mathbb{P}_{\boldsymbol{q}}^c$, the entries satisfy $F_{\boldsymbol{l}}(t,t+q'_{2i-2}|i,i-1) \geq 0$ only for the vectors $\boldsymbol{l}$ that give $q_{2i-2} =  q'_{2i-2}$, otherwise these entries must be zero for any $t$. Thus, for these vectors $\boldsymbol{l}$, according to Lemma~\ref{lemma addition}, we have 
\begin{equation}
    \begin{aligned}
        \sum_{\boldsymbol  l}f(\boldsymbol  l) F_{\boldsymbol  l} \left(t,t+q'_{2i-2}|i,i-1\right) &= \sum_{\boldsymbol  l \backslash q_{2i-2} \in \mathbb{Z}_{d}^{2m-3}} \mathcal{F}^{m}_{\boldsymbol  l}(t+c_{i})= \sum_{q_{2},\cdots,q_{2i-3},q_{2i-1}\in [d]}\mathcal{F}^{i}_{\boldsymbol  l}(t+c_{i}) = \kappa(t ,t+q'_{2i-2} |i,i-1) ,~~\forall t, i\geq 3, q'_{2i-2}\in[d].
    \end{aligned}
\end{equation}
For $i = 2$, the last equality is $\sum_{q_3}\mathcal{F}^{2}_{\boldsymbol  l}(t) = \kappa(t ,t+q'_{2} |21)$.
Similarly, for the input pair $(A_i, B_i)$ and a given $q'_{2i-1}\in [d]$, we also have
\begin{equation}
    \begin{aligned}
        \sum_{\boldsymbol  l} f(\boldsymbol  l) F_{\boldsymbol  l} \left(t,t-q'_{2i-1}|ii\right) = \sum_{\boldsymbol l \backslash q_{2i-1} \in \mathbb{Z}_{d}^{2m-3}} \mathcal{F}_{\boldsymbol{l}}^m(t+c_i)&= \sum_{q_2,\cdots,q_{2i-2}\in [d]} \mathcal{F}^{i}_{\boldsymbol  l}(t+c_{i}) = \kappa(t,t-q'_{2i-1}|ii),~~\forall t, i\geq 3, q'_{2i-1}\in[d].
    \end{aligned}
\end{equation}
For $i = 2$, we have $\sum_{q_2}\mathcal{F}^{2}_{\boldsymbol  l}(t) = \kappa(t,t-q'_3|22)$, which concludes the proof.
\end{proof}

\begin{definition}
    For a given integer $e\in[d]$, we define a subset of no-signaling behaviors $\boldsymbol{p}\in\mathcal{NS}(d,d,m,m)$ as $\mathbb{D}_{A_{x^*} = e}\subset \mathcal{NS}(d,d,m,m)$, in which the behaviors have deterministic marginal probabilities $p_A(a|x^*) = \delta_{a,e}$ on the measurement ${x^*}$. Here ${\rm Conv}\left( \cup_{e}\mathbb{D}_{A_{x^*} = e}\right)$ is also referred to a partially deterministic polytope~\cite{PhDthesis_2014_Eirk}. 
\end{definition}

\begin{proposition}\label{proposition 2}
For a given integer $e\in[d]$, any behavior $\boldsymbol{p}\in \mathbb{D}_{A_{x^*} = e} \subset \mathcal{NS}(d,d,m,m)$ admits a decomposition of
\begin{equation}
\label{eq:formofdecomposition}
        \boldsymbol{p} = \sum_{\boldsymbol  q} p(\boldsymbol  q) \boldsymbol{P}_{\boldsymbol{q}},
\end{equation}
where $\boldsymbol  q \coloneqq (q_1,q_2,\cdots,q_{2m-1})^{T}\in\mathbb{Z}_d^{2m-1}$. Here $p(\boldsymbol  q)\geq 0$, $\sum_{\boldsymbol  q }p(\boldsymbol  q) = 1$, $\sum_{\boldsymbol  q}$ means $\sum_{q_1,\cdots,q_{2m-1}\in [d]}$, and $\boldsymbol{P}_{\boldsymbol  q} \in \mathbb{P}_{\boldsymbol{q}}^c\subset \mathcal{P}(d,d,m,m)$. Thus, based on Lemma~\ref{lemma addition2}, no behavior in the partially deterministic polytope ${\rm Conv } \left( \cup_e \mathbb{D}_{A_{x^*} = e} \right) \subset {\rm Conv } \left(\cup_{\boldsymbol{q}}\mathbb{P}_{\boldsymbol{q}}^c \right)$ can violate any chain inequality Eq.~\eqref{eq:chain_manuscript}.
\end{proposition}

\begin{proof}
Firstly, for a given behavior $\boldsymbol{p}\in\mathbb{D}_{A_{x^*} = e}$, we define $\boldsymbol \kappa \coloneqq \mathrm{\Pi}_{\{A_1,\cdots,A_m\}\cup\{B_1,\cdots,B_m\}\backslash A_{x^*}} \boldsymbol{p} \in \mathcal{NS}(d,d,m-1,m)$. Then we have $p(ab|x^*y) = \delta_{a,e}\kappa_B(b|y)$, and the vector $\boldsymbol  \kappa\in \mathbb{R}^{t-md^2}$ is enough to determine $\boldsymbol{p}\in\mathbb{R}^{t}$, where $t = m^2d^2$. 
Then, we will construct the decomposition Eq.~\eqref{eq:formofdecomposition} based on Lemma~\ref{lemma 2}. Here we define $p(\boldsymbol  q)$ and $\boldsymbol{P}_{\boldsymbol{q}}$ based on $\boldsymbol \kappa$, which are
\begin{equation}
    p(\boldsymbol  q) \coloneqq p(q_1|\boldsymbol  l)f(\boldsymbol  l) ~~~\text{with}~~~p(q_1|\boldsymbol  l) \coloneqq F_{\boldsymbol  l}(e-q_1-q_2,e-q_1|21), ~~~\forall \boldsymbol{q} = (q_1,\boldsymbol{l})\in\mathbb{Z}_d^{2m-1},
\end{equation}
such that $\sum_{q_1\in[d]}p(q_1|\boldsymbol{l}) = 1$ for any $\boldsymbol{l} \in\mathbb{Z}_d^{2m-2}$. One can easily check that $p(\boldsymbol  q)\geq 0$ and $\sum_{\boldsymbol  q}p(\boldsymbol  q) = 1$. For $\mathrm{\Pi}_{\mathcal{I}^c}\boldsymbol{P}_{\boldsymbol{q}}$, we define them by 
\begin{equation}
\label{eq:B16}
    \mathrm{\Pi}_{\mathcal{I}^c\backslash \left(A_{x^*},\cdot\right)} \boldsymbol{P}_{\boldsymbol{q}} \coloneqq \mathrm{\Pi}_{\mathcal{I}^c\backslash \left(A_{x^*},\cdot\right)}\boldsymbol{F}_{\boldsymbol{l}},~~~P_{\boldsymbol{q}}(ab|x^*x^*) \coloneqq \delta_{a,e}\delta_{b,e-q_1}, ~~~P_{\boldsymbol{q}}(ab|x^*m) \coloneqq \delta_{a,e}\delta_{b,e-\| \boldsymbol{q}\|_1},~~~\forall  \boldsymbol{q} = (q_1,\boldsymbol{l})\in \mathbb{Z}_d^{2m-1}.
\end{equation}
Then there always exists $\boldsymbol{P}_{\boldsymbol{q}}\in \mathbb{P}_{\boldsymbol{q}}^c$ with the same projection as in Eq.~\eqref{eq:B16}. 

Next, we will prove that Eq.~\eqref{eq:formofdecomposition} holds for $p(\boldsymbol  q)$ and $\boldsymbol{P}_{\boldsymbol{q}}$ as defined above. For the input pairs $\mathcal{I}^c\backslash(A_{x^*},\cdot)$, we have
\begin{equation}
    \begin{aligned}
        \mathrm{\Pi}_{\mathcal{I}^c\backslash(A_{x^*},\cdot)}\boldsymbol{p} = \mathrm{\Pi}_{\mathcal{I}^c\backslash(A_{x^*},\cdot)}\boldsymbol{\kappa} = \mathrm{\Pi}_{\mathcal{I}^c\backslash \left(A_{x^*},\cdot\right)}\sum_{\boldsymbol l}f(\boldsymbol  l) \boldsymbol{F}_{\boldsymbol  l} = \mathrm{\Pi}_{\mathcal{I}^c\backslash \left(A_{x^*},\cdot\right)}\sum_{\boldsymbol q}{p}(\boldsymbol q) \boldsymbol{P}_{\boldsymbol q}.
    \end{aligned}
\end{equation}

For the input pair $(A_{x^*},B_{x^*})$, we have
\begin{equation}
    \begin{aligned}
    \label{eq:decompositionx1}
        \sum_{\boldsymbol  q} p(\boldsymbol  q)P_{\boldsymbol  q} (ab|x^*x^*) &= \sum_{q_1\in [d]}\sum_{\boldsymbol  l} p(q_1|\boldsymbol  l)f(\boldsymbol  l)\delta_{a,e}\delta_{b,e-q_1} = \delta_{a,e}\sum_{\boldsymbol  l} F_{\boldsymbol l}(b-q_2,b|21)f(\boldsymbol l) = \delta_{a,e}\sum_{\boldsymbol{l}}\sum_{a'\in [d]} \delta_{a',b-q_2}F_{\boldsymbol l}(a'b|21)f(\boldsymbol l),~~~\forall a,b.
    \end{aligned}
\end{equation}
Since $F_{\boldsymbol l}(a'b|21) = 0$ if $a' \neq b - q_2$, then Eq.~\eqref{eq:decompositionx1} is equal to $ \delta_{a,e}\kappa_B(b|x^*) = p(ab|x^*x^*)$. Thus, Eq.~\eqref{eq:formofdecomposition} also holds for $(A_{x^*},B_{x^*})$:
\begin{equation}
    \mathrm{\Pi}_{(A_{x^*},B_{x^*})}\boldsymbol{p} = \mathrm{\Pi}_{(A_{x^*},B_{x^*})}\sum_{\boldsymbol q}{p}(\boldsymbol q) \boldsymbol{P}_{\boldsymbol q}.
\end{equation}

Finally, for the input pair $(A_{x^*},B_{m})$, we have
\begin{equation}
\label{eq:B20}
        \sum_{\boldsymbol  q} p(\boldsymbol  q)P_{\boldsymbol  q} (ab|x^*m) = \delta_{a,e}\sum_{q_1 \in [d]}\sum_{\boldsymbol  l} p(q_1|\boldsymbol  l)f(\boldsymbol  l)\delta_{b,e-\|\boldsymbol  q\|_1} = \delta_{a,e}\sum_{\boldsymbol  l} f(\boldsymbol  l) F_{\boldsymbol  l}(b+\|\boldsymbol  l\|_1 - q_2,b + \|\boldsymbol  l\|_1 |21).
\end{equation}
Meanwhile, based on Lemma~\ref{lemma 2}, we have
\begin{equation}
    \begin{aligned}
        \kappa_B(b|m) & = \sum_{a'}\kappa(a'b|mm) = \sum_{a'}\sum_{\boldsymbol{l}}f(\boldsymbol{l})F_{\boldsymbol{l}}(a'b|mm) = \sum_{\boldsymbol{l}}f(\boldsymbol{l})F_{\boldsymbol{l}}(b+q_{2m-1},b|mm) = \sum_{\boldsymbol{l}}\mathcal{F}_{\boldsymbol{l}}^m(b+q_{2m-1}+c_m). 
    \end{aligned}
\end{equation}
Since $\mathcal{F}_{\boldsymbol{l}}^m(b+q_{2m-1}+c_m) = \mathcal{F}_{\boldsymbol{l}}^m\left((b+\|\boldsymbol{l}\|_1 - q_2) + c_2\right) $, we have 
\begin{equation}
    \kappa_B(b|m) = \sum_{\boldsymbol  l} f(\boldsymbol  l) F_{\boldsymbol  l}(b+\|\boldsymbol  l\|_1 - q_2,b + \|\boldsymbol  l\|_1 |21).
\end{equation}
Then Eq.~\eqref{eq:B20} is equal to $\delta_{a,e}\kappa_B(b|m)=p(ab|x^*m)$, thereby
\begin{equation}
    \mathrm{\Pi}_{(A_{x^*},B_{m})}\boldsymbol{p} = \mathrm{\Pi}_{(A_{x^*},B_{m})}\sum_{\boldsymbol q}{p}(\boldsymbol q) \boldsymbol{P}_{\boldsymbol q}.
\end{equation}
Thus, we have 
\begin{equation}
    \mathrm{\Pi}_{\mathcal{I}^{c}}        \boldsymbol{p} = \mathrm{\Pi}_{\mathcal{I}^{c}}\sum_{\boldsymbol  q} p(\boldsymbol  q) \boldsymbol{P}_{\boldsymbol{q}}.
\end{equation}
Based on the Definition~\ref{definition1}, we can easily obtain a convex decomposition in the form of Eq.~\eqref{eq:formofdecomposition}, which concludes the proof. 
\end{proof}

\setcounter{equation}{0}

\renewcommand\theequation{C\arabic{equation}}
\section{Proof of Result 4}
\label{appC}
In Result~\ref{result 4} of the manuscript, we consider cases that take input distribution into account and provide a necessary and sufficient condition for randomness certification. Firstly, we change the target function of Eq.~\eqref{eq:definition_steer_rand} in the manuscript to
\begin{equation}
    \begin{aligned}
    \label{eq:C1}
        P^{\rm S}_{\text {guess }} \left(X\right) & = \max_{\left\{\sigma_{a \mid x}^{\gamma}\right\}} \sum_{x\in X} p(x) \sum_{\gamma}\operatorname{Tr}\left(\sigma_{a=\gamma(x) \mid x}^\gamma\right),
    \end{aligned}
\end{equation}
where $X\subseteq \{1,2,\cdots,m\}$, $p(x)>0$ denotes the relative frequencies of inputs, and $\sum_{x\in X}p(x) = 1$. Here Eve's guess for different measurements is a string $\gamma = (a_{x=x_1},\cdots,a_{x = x_s})$, where $\gamma(\cdot)$ being a function from $X$ to $[d]$. 

Consequently, we demonstrate that given an assemblage $\{\sigma_{a|x}^{\rm obs}\}$, the guessing probability satisfies $P^{\rm S}_{\text{guess}} \left(X\right) < 1$ iff it lies outside the set $\mathcal{R}_{X}^{\rm S}\coloneqq\left\{\sigma_{a|x}: \sum_a\sigma_{a|x} = \sum_a \sigma_{a|x'},~\forall x,x',~~\mathcal{C}_{A} \supseteq K_{s,m-s} (X)\right\}$ for any $X\subseteq\{1,\cdots,m\}$. The structure $K_{s,m-s} (X)$ implies the set of measurements $X$, when combined with every other measurement $x\notin X$, are compatible (e.g., Fig.~\ref{fig:structure}~(c) in the manuscript).
The detailed proof, generalized from Appendix~\ref{appA}, is provided below.

\begin{proof}
For the ``only if'' part, if $\mathcal{C}_A \supseteq K_{s,m-s}\left(X\right)$, there exists a parent POVM of the set of measurements $X$, denoted as $\mathcal{X}$ with elements $M_{\boldsymbol  a| \mathcal{X}}$, where $\boldsymbol  a \coloneqq (a_{x_1},a_{x_2},\cdots,a_{x_{s}} ) \in \mathbb{Z}_d^{s}$. This parent POVM has $d^s$ outcomes, such that it reduces to a single measurement $x \in X$ by classical post-processing~\cite{rmp_2023_otfried}. As an example, 
\begin{equation}
    M_{a|x} = \sum_{x'\in X\backslash x}\sum_{a_{x'}\in[d]}M_{\boldsymbol  a = (a_{x_1},\cdots,a_x = a,\cdots,a_{x_s})| \mathcal{X}}, \quad \forall x\in X.
\end{equation}
The compatibility structure of measurement $\mathcal{X}$ together with the other measurements $x\notin X$ is described by the star graph $K_{1,m-s}(\mathcal{X})$. Therefore, for a given bipartite state $\rho_{AB}$, the corresponding ``$K_{1,m-s}(\mathcal{X})$-partially-unsteerable'' assemblage can be treated as coming from $m - s + 1$ measurements: 
\begin{equation}
    \begin{aligned}
    \label{eq:C3}
    \sigma_{a|x}^{\text{obs}} &= \sum_{x'\in X\backslash x}\sum_{a_{x'}\in[d]}  \theta_{\boldsymbol  a = (a_{x_1},\cdots,a_x = a,\cdots,a_{x_s})|\mathcal{X}}^{\text{obs}}\quad \text{with} \quad \theta_{\boldsymbol  a|\mathcal{X}}^{\text{obs}} = \sum_{\mu}D_{\text{Local}}(\boldsymbol  a|\mathcal{X},\mu)\tau_\mu^{i}, \quad \forall a, x\in X, e_i\in E\left(K_{1,m-s}(\mathcal{X})\right),\\
    \sigma_{a|x}^{\text{obs}} &= \sum_{\mu}D_{\text{Local}}(a|x,\mu)\tau_\mu^{i},\quad \forall a,x\notin X, \jt{\quad} (\mathcal{X},x)= e_i\in E\left(K_{1,m-s}(\mathcal{X})\right),
    \end{aligned}
\end{equation}
where $\tau_\mu^i = \text{Tr}_A\left[\left(\mathcal{G}_{\mu}^i\otimes \mathbb{I}^B\right)\rho_{AB}\right]$ are local hidden states, and $e_i$ is the edge of star graph $E\left(K_{1,m-s}(\mathcal{X})\right)$. Still, for the sake of simplicity, we will use $e_i$ to denote the edge of $E\left(K_{1,m-s}(\mathcal{X})\right)$ in the following. Here $\left\{\mathcal{G}_{\mu}^i\right\}_{\mu}$ is the parent POVM of measurements $\{ \mathcal{X}, x\}$, where $(\mathcal{X},x) = e_i$. We note that $\theta_{\boldsymbol  a|\mathcal{X}}^{\text{obs}} = \text{Tr}_{A}\left[ \left( M_{\boldsymbol  a|\mathcal{X}} \otimes \mathbb{I}^B\right)\rho_{AB}  \right] $ in Eq.~\eqref{eq:C3} are the same for all $i\in\{1,\cdots,m-s\}$. Similar to Appendix~\ref{appA}, we can directly construct an optimal solution such that $P_{g}^{\rm S}(X) = 1$ in Eq.~\eqref{eq:C1}:
\begin{equation}
\begin{aligned}
\label{eq:C4}
    \xi_{a|x}^{\gamma} &\coloneqq \sum_{x'\in X\backslash x}\sum_{a_{x'}\in[d]}\theta_{\boldsymbol  a = (a_{x_1},\cdots,a_x = a, \cdots, a_{x_s})|\mathcal{X}}^{\gamma},\quad \text{with} \quad \theta_{\boldsymbol  a|\mathcal{X}}^\gamma \coloneqq \sum_{\mu\in \Lambda^{(\gamma)}}D_{\text{Local}}(\boldsymbol  a|\mathcal{X},\mu)\tau_\mu^{i}, \quad  \forall a, x\in X, \gamma,\\
    \xi_{a|x}^{\gamma} &\coloneqq \sum_{\mu\in\Lambda^{(\gamma)}}D_{\text{Local}}(a|x,\mu) \tau_\mu^i, \quad \forall a, x\notin X, (\mathcal{X},x) = e_i,\gamma,
\end{aligned}
\end{equation}
where $\Lambda^{(\gamma)} \coloneqq \{\mu|\mu(\mathcal{X}) = \gamma \}$. 
This construction satisfies
\begin{equation}
   \theta_{\boldsymbol  a|\mathcal{X}}^{\gamma = \boldsymbol  a}  = \sum_{\gamma} \theta_{\boldsymbol  a|\mathcal{X}}^{\gamma} = \theta_{\boldsymbol  a|\mathcal{X}}^{\text{obs}}, ~~\text{and}~~ {\rm Tr}\left( \theta_{\boldsymbol a|\mathcal{X}}^{\gamma \neq \boldsymbol{a}} \right) = 0, ~~~\forall \boldsymbol{a},
\end{equation}
Since $\theta_{\boldsymbol{a}|\mathcal{X}}^{\rm obs}$ are independent of $i$, the constructions $\theta_{\boldsymbol a | \mathcal{X}}^\gamma$ and $\xi_{a|x\in X}^\gamma$ are also independent of $i$. 
For $a \neq \gamma(x)$, we have $\theta_{\boldsymbol  a|\mathcal{X}}^{\gamma} = 0$ for any $\boldsymbol{a}\backslash a_x$ and $\xi_{a|x}^{\gamma} = 0$. With this construction, Eve can always guess the outcomes of measurement $x\in X$ correctly:
\begin{equation}
    \begin{aligned}
    \sum_\gamma\text{Tr}\left(\xi_{a=\gamma(x)|x}^\gamma\right) = \sum_\gamma\text{Tr}\left(\theta_{\boldsymbol  a = \gamma|\mathcal{X}}^\gamma\right) =  \sum_\gamma\text{Tr}\left(\theta_{\boldsymbol  a = \gamma|\mathcal{X}}^{\text{obs}}\right) = 1, \quad \forall x\in X,
    \end{aligned}
\end{equation}
which gives $P_{g}^{\rm S}(X) = 1$ in Eq.~\eqref{eq:generalization}. 
Finally, one can easily check that this construction aligns with the observed assemblage. Based on Eq.~\eqref{eq:C4}, it also satisfies the no-signaling conditions
\begin{equation}
    \begin{aligned}
         \sum_a \xi_{a|x}^\gamma &= \sum_{\boldsymbol  a} \theta_{\boldsymbol  a|\mathcal{X}}^\gamma = \theta_{\boldsymbol  a = \gamma|\mathcal{X}}^{\gamma} =\sum_{\mu \in \Lambda^{(\gamma)}} \tau_{\mu}^i = \sum_a\xi_{a|x'}^{\gamma}, \quad \forall \gamma,  x'\notin X, x\in X, (\mathcal{X},x
        ') = e_i,\\
         \sum_a \xi_{a|x}^\gamma &= \sum_{\boldsymbol  a} \theta_{\boldsymbol  a|\mathcal{X}}^\gamma =  \sum_a \xi_{a|x'}^\gamma, \quad \forall \gamma,\quad x,x'\in X. 
    \end{aligned}
\end{equation}
 Therefore, when $\mathcal{C}_A \supseteq K_{s,m-s}\left(X\right)$, we have $P_{g}^{\rm S}(X) = 1$ in Eq.~\eqref{eq:generalization} in the manuscript. 

For the ``if'' part, $P_{g}^{\rm S}(X) = 1$ only when there exists a solution $\left\{ \xi_{a|x}^\gamma \right\}$ such that it satisfies all the constraints and
\begin{equation}
    \begin{aligned}
        \sum_\gamma \text{Tr}\left( \xi_{a = \gamma(x)|x}^\gamma\right) = 1, \quad \forall x\in X.
    \end{aligned}
\end{equation}
This indicates that 
\begin{equation}
\label{eq:C8}
    \text{Tr}\left( \xi_{a |x}^{\gamma} \right) = 0, \quad \forall \gamma,  a \neq  \gamma(x),x \in X,
\end{equation}
 as $\sum_{a,\gamma} \text{Tr}\left( \xi_{a |x}^{\gamma} \right) = 1$ and $\text{Tr}\left( \xi_{a |x}^{\gamma} \right)\geq 0$. 
For each $\gamma$, the no-signaling condition ensures that the ensemble admits quantum realization. Thus, there exist a state $\rho_{AB}^\gamma$ and POVMs $\left\{\mathcal{M}_{a|x}^\gamma\right\}_{a,x}$ such that
\begin{equation}
    \begin{aligned}
        \xi_{a|x}^\gamma = p(\gamma)\text{Tr}_A\left[\left(\mathcal{M}_{a|x}^\gamma \otimes \mathbb{I}^B\right) \rho_{AB}^\gamma\right], \quad \forall a,\gamma, x\in\left\{1,\cdots,m\right\},
    \end{aligned}
\end{equation}
where $p(\gamma) = \sum_a \text{Tr}\left( \xi_{a|x}^\gamma \right)$. 
Consider that Eve gives a guess $\gamma$ by implementing one measurement on her side. The ensemble can then be explained by
\begin{equation}
    \begin{aligned}
        \xi_{a|x}^\gamma = \text{Tr}_{E\boldsymbol{A}}\left[\left( \mathcal{M}_{\gamma} ^E \otimes \mathcal{M}_{a|x} ^{\boldsymbol{A}} \otimes \mathbb{I}^B\right) \rho_{E\boldsymbol{A}B}\right],  \quad \forall a,\gamma, x\in\left\{1,\cdots,m\right\}.
    \end{aligned}
\end{equation}
where $\rho_{E\boldsymbol{A}B} = \sum_\gamma p(\gamma) |\gamma \gamma \rangle _{EA'} \langle \gamma \gamma | \otimes \rho_{AB}^\gamma$, $\mathcal{M}_{\gamma}^E = |\gamma\rangle _{E} \langle \gamma|$, and $\mathcal{M}_{a|x}^A = \sum_{\gamma}|\gamma\rangle _{A'} \langle \gamma| \otimes \mathcal{M}_{a|x}^\gamma$. Here $|\gamma\rangle = |\gamma(x_1),\cdots,\gamma(x_s)\rangle$ are orthogonal states. Based on Eq.~\eqref{eq:C8}, we have
\begin{equation}
    \begin{aligned}
        \sigma_{a|x}^{\text{obs}} &= \sum_{\gamma }\xi_{a |x}^\gamma = \sum_{\gamma \in \left\{ \gamma: \gamma(x) = a\right\}}\xi_{a|x}^\gamma = \sum_{\gamma \in \left\{ \gamma: \gamma(x) = a\right\}} \sum_{a'}\xi_{a'|x}^\gamma = \text{Tr}_{EA}\left[\left( | a\rangle_{E_x} \langle a| \otimes \mathbb{I}^{\boldsymbol{A}}\otimes \mathbb{I}^B\right) \rho_{E\boldsymbol{A}B}\right], \forall x\in X, \gamma,
        \\
        \sigma_{a|x}^{\text{obs}} &= \sum_{\gamma }\xi_{a |x}^\gamma = \text{Tr}_{E\boldsymbol{A}}\left[\left( \mathbb{I}^E \otimes \mathcal{M}_{a|x}^A\otimes \mathbb{I}^B\right) \rho_{E\boldsymbol{A}B}\right],~~~\forall a,\gamma, x\notin X.
    \end{aligned}
\end{equation}
Therefore, by treating $E\boldsymbol{A}|B$ as bipartition, when the subsystem $E\boldsymbol{A}$ receives input $x\in X$, Eve performs measurement $\left\{|a\rangle _{E_x}\langle a|\right\}_a$. When the subsystem $E\boldsymbol{A}$ receives input $x\notin X$, Alice performs measurements $\left\{\mathcal{M}_{a|x}^{\boldsymbol{A}}\right\}$, thereby the compatibility structure of the subsystem $E\boldsymbol{A}$ can be explained by the hypergraph $K_{s,m-s}(X)$, which concludes the proof.
\end{proof}

\begin{remark}
    Given an assemblage $\{\sigma_{a|x}^{\rm obs}\}$, when $s = m$, the guessing probability satisfies $P_g^{\rm S}(X) = 1$ iff $\{\sigma_{a|x}^{\rm obs}\}$ is unsteerable.
\end{remark}
\begin{proof}
The ``if'' direction is straightforward. 
For the ``only if'' direction, if $P_g^{\rm S}(X) = 1$ in Eq.~\eqref{eq:C1} in the manuscript and $s = m$ such that $X = \left\{1,2,\cdots,m\right\}$, based on Eq.~\eqref{eq:C8}, for each $\gamma$, we have 
\begin{equation}
    \begin{aligned}
        \sum_a \xi_{a|x}^\gamma =  \sum_a \xi_{a|x'}^\gamma ~~ \Rightarrow ~~ \xi_{a = \gamma(x)|x}^\gamma =  \xi_{a = \gamma(x')|x'}^\gamma , \quad \forall x, x'\in X, \gamma,
    \end{aligned}
\end{equation}
as the solution $\left\{\xi_{a|x}^\gamma\right\}_{a,x}$ satisfies the no-signaling conditions. 
We can simply consider an ensemble that is independent of the input and output
\begin{equation}
    \zeta^\gamma = \xi_{a = \gamma(x_1)|x_1}^\gamma  = \xi_{a = \gamma(x_2)|x_2}^\gamma = \cdots = \xi_{a = \gamma(x_m)|x_m}^\gamma, \quad \forall \gamma,
\end{equation}
where $\sum_\gamma\text{Tr}\left( \zeta^\gamma \right) = \sum_\gamma\text{Tr}\left( \xi_{a = \gamma(x_1)|x_1}^\gamma\right) = \sum_\gamma\sum_{a'}\text{Tr}\left( \xi_{a'|x_1}^\gamma\right) = 1$. 
Thus, the observed assemblage can be written as 
\begin{equation}
    \begin{aligned}
        \sigma_{a|x}^{\text{obs}} = \sum_\gamma \xi_{a|x}^\gamma = \sum_{\gamma}D_{\text{Local}}(a|x,\gamma) \zeta^\gamma, \quad \forall a,x\in \left\{1,\cdots,m\right\}.
    \end{aligned}
\end{equation}
Therefore, the observed assemblage admits the LHS model, and $\{\zeta^{\gamma}\}$ are the hidden ststaes.
\end{proof}

\setcounter{figure}{0}
\renewcommand\thefigure{D\arabic{figure}}
\section{Steering weight and steering-based randomness}
\label{appD}
The steering weight was proposed to quantify the steering~\cite{prl_2014_paul_sw}. For a given assemblage $\sigma_{a|x}^{\rm obs}$, it can always be decomposed as a convex combination of an LHS assemblage $\sigma_{a|x}^{\rm LHS}$ and a generic no-signaling assemblage $\gamma_{a|x}$, i.e.,
\begin{equation}
\label{eq:D1}
    \sigma_{a|x}^{\rm obs} = p\gamma_{a|x} + (1-p)\sigma_{a|x}^{\rm LHS}, ~~~\forall a,x,
\end{equation}
where $p\in [0,1]$.
Then, the steering weight is the minimum of $p$ over all possible decompositions of Eq.~\eqref{eq:D1}:
\begin{equation}
\begin{aligned}
\label{eq:D2}
\operatorname{SW}\left(\sigma_{a \mid x}^{\rm obs}\right)= & \min _{\left\{\gamma_{a \mid x}\right\},\left\{\sigma_\lambda\right\}, p}  p, \\
\text { s.t. } & \sigma_{a \mid x}^{\rm obs}=p \gamma_{a \mid x}+(1-p) \sigma_{a \mid x}^{\mathrm{LHS}}, \quad \sigma_{a \mid x}^{\mathrm{LHS}}=\sum_\lambda D(a | x, \lambda) \sigma_\lambda, \quad \forall a, x, \\
& \sum_a \gamma_{a \mid x}=\sum_a \gamma_{a \mid x^{\prime}}, \quad \operatorname{Tr} \sum_a \gamma_{a \mid x}=1,  \quad \forall x, x^{\prime}, \\
& \operatorname{Tr} \sum_\lambda \sigma_\lambda=1, \quad \gamma_{a \mid x}, \sigma_\lambda \succeq 0, \quad \forall a, x,\lambda,
\end{aligned}
\end{equation}
which can be computed via an SDP~\cite{rpp_2016_paul}. 
For instance, when the given assemblage is extreme but very close to the unsteerable set, this will give a very small steering robustness but a maximum steering weight, e.g., by performing Pauli measurements on the partially entangled states $|\Psi_{\theta}\rangle$ shown in the manuscript, the resulting assemblage gives maximum steering weight and certifies maximum steering-based randomness for any $\theta\in(0,\pi/4]$. 
Thus, it is natural to discuss the relation between the amount of steering-based randomness and steering weight. For a given assemblage $\sigma_{a|x}^{\rm obs}$, the exact value of guessing probability $P_g^{\rm S}(x^*)$ can be given.
Based on the optimal solution found in Eq.~\eqref{eq:D2}, we can also provide a lower bound of the guessing probability $P^{\rm S}_g(x^*)$:
\begin{equation}
\label{eq:C3}
    P^{\rm S}_g(x^*) \geq \operatorname{SW}\left(\sigma_{a \mid x}^{\rm obs}\right) \max_a {\rm Tr} \left( \gamma_{a|x^*} \right) + 1 - \operatorname{SW}\left(\sigma_{a \mid x}^{\rm obs}\right).
\end{equation}
This means ${\rm SW}(\sigma_{a|x}^{\rm obs}) = 0$ gives $P_g^{\rm S}(x^*) = 1$. Since not all steerable assemblages can certify nonzero randomness, we can provide a tighter bound by considering the partially unsteerable assemblages. Thus, we define the partially steering weight as 
\begin{equation}
\begin{aligned}
\label{eq:D2}
\operatorname{PSW}\left(\sigma_{a \mid x}^{\rm obs}\right)= & \min _{\left\{\gamma_{a \mid x}\right\},\left\{\sigma_\lambda^i\right\}, p}  p, \\
\text { s.t. } & \sigma_{a \mid x}^{\rm obs}=p \gamma_{a \mid x}+(1-p) \sigma_{a \mid x}^{\mathrm{PLHS}}, \quad \forall a, x,\\
& \sigma_{a \mid x}^{\mathrm{PLHS}} =\sum_\lambda D(a | x, \lambda) \sigma_\lambda^i, \quad \forall a, \quad (x^*,x) = e_i \in E\left( K_{1,m-1}(x^*) \right), \\
& \sum_a \gamma_{a \mid x}=\sum_a \gamma_{a \mid x^{\prime}}, \quad \operatorname{Tr} \sum_a \gamma_{a \mid x}=1, \quad \forall x, x^{\prime}, \\
&\operatorname{Tr} \sum_\lambda \sigma_\lambda^i=1, \quad \sigma_\lambda^i \succeq 0, \quad \forall \lambda,i, \quad \gamma_{a \mid x} \succeq 0, \quad \forall a, x.
\end{aligned}
\end{equation}
Similarly, one can provide a lower bound of the guessing probability $P^{\rm S}_g(x^*)$:
\begin{equation}
\label{eq:C5}
    P^{\rm S}_g(x^*) \geq \operatorname{PSW}\left(\sigma_{a \mid x}^{\rm obs}\right) \max_a {\rm Tr} \left( \gamma_{a|x^*} \right) + 1 - \operatorname{PSW}\left(\sigma_{a \mid x}^{\rm obs}\right).
\end{equation}
For instance, consider the states $\rho_{AB}^{p,\theta}$ discussed in the manuscript. When Alice performs three Pauli measurements $\{\hat{X},\hat{Y},\hat{Z}\}$, the resulting guessing probability and its lower bounds given by Eq.~\eqref{eq:C3} and Eq.~\eqref{eq:C5} are shown in Fig.~\ref{fig:D1}.
\begin{figure}[h]
\centering
\includegraphics[width=0.5\textwidth]{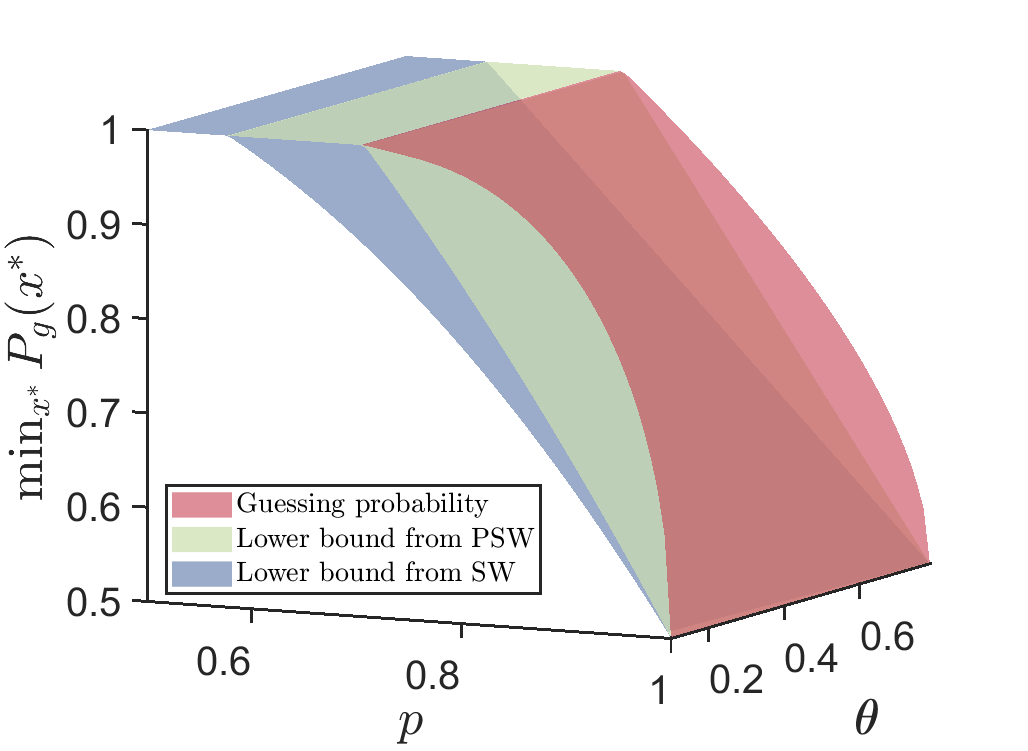}
\caption{\label{fig:D1} The guessing probability and its lower bounds certified from the state $\rho_{AB}^{p,\theta}$. The red surface represents the exact value of $\min_{x^*}P_g^{\rm S}(x^*)$. The green and blue surfaces represent the lower bounds of $P_g^{\rm S}(x = {\rm argmin}_{x^*} P_g^{\rm S}(x^*))$ given by Eq.~\eqref{eq:C5} and Eq.~\eqref{eq:C3} respectively. For the same $\{\sigma_{a|x}^{\rm obs}\}$, the lower bound given by Eq.~\eqref{eq:C5} is tighter than the lower bound given by Eq.~\eqref{eq:C3}. Further, the thresholds for nonzero certifiable randomness on the green surface are the same as that of the red surface, which is consistent with Result~\ref{result 1}. }
\end{figure}

\end{appendix}

\end{document}